\documentclass[10pt,twocolumn,twoside]{IEEEtran}

\usepackage{cite}
\usepackage{citesort}
\usepackage{graphicx}
\usepackage{amsmath}
\usepackage{times}
\usepackage{subfigure}
\usepackage{latexsym}
\usepackage{graphicx}
\usepackage{bm}
\usepackage{amssymb}
\usepackage[center]{caption2}
\usepackage{stfloats}
\usepackage{cases}
\usepackage{array}
\usepackage{setspace}
\usepackage{fancyhdr}

\usepackage[noend]{algpseudocode}
\usepackage{algorithmicx,algorithm}

\newlength{\oldparindent}
\newlength{\oldparskip}
\renewcommand{\QED}{\QEDopen}
\newtheorem{theorem}{Theorem}

\newtheorem{corollary}{Corollary}

\newtheorem{remark}{Remark}

\renewcommand{\baselinestretch}{0.95}

\algdef{SE}[SUBALG]{Indent}{EndIndent}{}{\algorithmicend\ }%
\algtext*{Indent}
\algtext*{EndIndent}

\begin{document}
\title{Reconsidering Design of Multi-Antenna NOMA Systems with Limited Feedback}
\author{Zhiyao~Tang, Liang~Sun, \emph{Member}, \emph{IEEE}, Lu~Cao, Shutong Qi, and Yong~Feng,
\emph{Member}, \emph{IEEE}
\thanks{ Z. Tang, L. Sun, L. Cao and Shutong Qi are with the School of Electronic and Information Engineering, Beihang University, Beijing, China (email: \{zytang\}@buaa.edu.cn,
\{eelsun\}@buaa.edu.cn, caolu66@buaa.edu.cn, and tony0612@buaa.edu.cn). L. Sun is also with International Research Institute for Multidisciplinary Science, Beihang University, Beijing.}
\thanks{Y. Feng is with the Yunnan Key Laboratory of Computer Technology Applications, Kunming University of Science and Technology, China (email: fybraver@kmust.edu.cn).}
}

\IEEEaftertitletext{\vspace{-0.75\baselineskip}}

\maketitle

%-----------------------------------------------------------------------------------------
% Abstract
%-----------------------------------------------------------------------------------------

\begin{abstract}
We provide in this paper a comprehensive solution to the design, performance analysis, and optimization of a multi-antenna non-orthogonal multiple access (NOMA) system for multiuser downlink communications under a general limited channel state information (CSI) feedback framework for frequency division duplex mode. We design a general framework including user clustering, joint power and bits allocation, CSI quantization and feedback, signal superposition coding, transmit beamforming, and successive interference cancellation at receivers. Then, we conduct a mathematically strict performance analysis of the considered system, and obtain a closed-form lower bound on the ergodic rate of each user in terms of transmit power, CSI quantization accuracy and channel conditions. For exploiting the potentials of multiple-antenna techniques in NOMA systems, we jointly optimize two key parameters, i.e., transmit power and the number of feedback bits allocated to each user, and propose low-complexity closed-form solutions.
Moreover, through asymptotic analysis, we reveal the interactions between the main system parameters and their impacts on the joint power and feedback bits allocation result, and hence show some guidelines on the system design. Finally, numerical results validate the correctness of our theoretical analysis and demonstrate the advantages of the proposed algorithms over the most related state of the art.
\end{abstract}

\begin{keywords}
Non-orthogonal multiple access, multiple-antenna techniques, limited feedback, random vector quantization, performance
analysis
\end{keywords}

\section{Introduction}\label{sec:introduction}
Non-orthogonal multiple access (NOMA) has recently attracted considerable attentions as a promising technique in the fifth generation (5G) mobile networks due to its potential in achieving high spectral efficiency \cite{Y.Saito2013,B.Kim_downlink_NOMA}. By combining superposition coding at the transmitters with successive interference cancellation (SIC) at the receivers, all the users in a NOMA system are able to decode their desired signal even though they share the same frequency, time and space resources. The spectral efficiency improvement of NOMA technique over orthogonal multiple access (OMA) has been well investigated for single-antenna systems \cite{Ding_NOMA2014,Timotheou_NOMA2015,J.Choi_NOMA2016,CLWANG_NOMA2016,Ding.UserPairing2016,Y.Zhang_NOMA2016}.
%Specifically, the performance of NOMA in a downlink cellular system with randomly distributed users was studied in \cite{Ding_NOMA2014}. The power allocation of NOMA system considering user fairness was studied in \cite{Timotheou_NOMA2015}.
%The impact of user pairing on the performance in NOMA system was shown in \cite{Ding.UserPairing2016}.
The spatial multiplexing property of multi-antenna systems can be naturally combined with power-domain NOMA to further enhance the spectral efficiency
\cite{JingjingCui_NOMA,Choi_NOMA2016,ZhiguoDing2016,YI.Choi2017,Z.Chen_NOMA2016,Z.Chen_NOMA2017,Nguyen_NOMA2017,C.Chen_NOMA,Q.Zhang_NOMA,YI.Choi2017,Ding_NOMA2016,Liu2016,Q.Yang_2017,Ding_NOMA2016,Liu2016,Q.Yang_2017,X.Chen_2017}.
%%%%%%%%%%%%%%%%%%%%%%%%%%%%%%%%%%%%%%%%%%%%%%%%%%
%It was noted in \cite{L.Zhang2017} that NOMA can also be regarded as a spectrum sharing technique.
NOMA was also regarded as a spectrum sharing technique in \cite{L.Zhang_JSAC_2017}.

It was noted in \cite{JingjingCui_NOMA} that all existing works on multi-antenna NOMA systems can be classified into two categories: 1) there is a cluster of multiple users supported by each channel spatial dimension (or by a beamforming (BF) vector) in average that perform NOMA independently with the other clusters \cite{JingjingCui_NOMA,Choi_NOMA2016,ZhiguoDing2016,YI.Choi2017,Z.Chen_NOMA2016,Z.Chen_NOMA2017,Nguyen_NOMA2017}, and 2) there is only one user supported by each channel spatial dimension in average\footnote{The transmission of this category is the same as that in OMA systems with SIC receiver employed by some users.}\cite{C.Chen_NOMA,Q.Zhang_NOMA}.
%1) each channel spatial dimension provided by a BF vector supports a cluster of multiple users performing NOMA independently \cite{B.Kim_downlink_NOMA,ZhiguoDing2016,YI.Choi2017,JingjingCui_NOMA} and 2) each channel spatial dimension supports only one user\footnote{The transmission of this category is the same as that in OMA systems and SIC receiver is employ by some users.}\cite{C.Chen_NOMA,Q.Zhang_NOMA}.
For the first category, \cite{B.Kim_downlink_NOMA} proposed a BF-based NOMA scheme, where a user clustering and power allocation algorithm was proposed to reduce the inter-cluster and intra-cluster interference.
%Considering two imperfect CSI models with channel estimation uncertainty and channel distribution information (CDI) only,
\cite{JingjingCui_NOMA} proposed a NOMA scheme with joint the power allocation and BF vectors design to maximize the system utility subjected to the probabilistic constraints. Following the same line of \cite{B.Kim_downlink_NOMA}, \cite{ZhiguoDing2016} proposed a new design of precoding and detection matrices for more general multi-antenna NOMA system with multi-antenna users.
\cite{YI.Choi2017} purposed a NOMA user grouping and clustering method based on a long-term feedback for large-scale multi-antenna systems along with user scheduling based on a short-term feedback.
%\cite{YI.Choi2017} also purposed multi-antenna NOMA scheme for large-scale multiuser MIMO systems with user grouping and clustering based on a long-term feedback along with user scheduling based on a short-term feedback.
A theoretical framework of the novel concept termed quasi-degradation first introduced in \cite{Z.Chen_NOMA2016} was fully studied in \cite{Z.Chen_NOMA2017}, where a closed-form hybrid-NOMA precoding algorithm was obtained for two-user MISO-NOMA systems. \cite{Nguyen_NOMA2017} designed linear precoders for cooperative signal suppositions at multiple base stations (BSs) in multi-cell multi-antenna NOMA systems to maximize the total throughput.
Obviously, the second category does not fully explore the spatial multiplexing capability of multi-antenna systems.

The quality of CSI available at transmitter (CSIT) and receiver of a multi-antenna system plays an important role and determines the performance of any multiple-antenna technique. In practice, the CSI acquisition methods can be mainly classified into the following two kinds. In time duplex division (TDD) systems, by employing the channel reciprocity of downlink and uplink channels, the transmitter can obtain the downlink CSI through channel estimation during the uplink training. While in frequency division duplex (FDD) systems, the downlink CSI is usually first estimated and quantized using a pre-determined codebook at each receiver, and then conveyed to transmitter via feedback channel. Specifically, the NOMA schemes in \cite{C.Chen_NOMA,Q.Zhang_NOMA,JingjingCui_NOMA} were designed based on the kind of CSI acquisition in TDD mode, whilst the NOMA schemes in \cite{Ding_NOMA2016,Liu2016,Q.Yang_2017} were designed based on the kind of CSI acquisition in FDD mode. Particularly, the work in \cite{X.Chen_2017} had tried to employ a unified channel model to embrace both TDD and FDD modes. For both practical TDD and FDD modes, there will be residual inter-cluster and intra-cluster interference caused by the imperfect CSIT.

To the best of our knowledge, there have been very limited previous works on the transmit-receive design, performance analysis and optimization for the multi-antenna NOMA systems with quantized CSIT through limited feedback, with the very few exceptions in \cite{Ding_NOMA2016,Liu2016,Q.Yang_2017,X.Chen_2017}.
Specifically, the outage performance of NOMA was investigated in \cite{Ding_NOMA2016} for each group decomposed from a massive-antenna system with one-bit feedback. \cite{Liu2016} investigated the traditional zero-forcing BF (ZFBF) and random BF technologies for the downlink NOMA systems. A joint user selection and power allocation scheme was proposed to reduce the multiuser (MU) interference and improve the sum-rate performance based on both quantized channel vector and perfect SINR feedback. The performance was evaluated and analyzed through simulations.
Also employing ZFBF, \cite{Q.Yang_2017} proposed a dynamic user scheduling and clustering strategy and considered the net throughout as metric based on an approximated closed-form analytical outage probability for the delay-intolerant systems. Based on ZFBF, \cite{X.Chen_2017} tried to propose a similar multi-antenna NOMA scheme to that in \cite{Q.Yang_2017} with a unified imperfect CSIT model applicable to both FDD and TDD systems. The authors in \cite{X.Chen_2017} had also tried to jointly optimize the transmit power and feedback bits of each users. However, \cite{X.Chen_2017} did not use the well recognized practical CSI model for FDD systems as in \cite{Jindal06,Yoo_limited,Liu2016,Q.Yang_2017}.
One can see the model used by \cite{X.Chen_2017} actually assumed some terms in the practical CSI model that should vary at the same speed of channel fading to be constant\footnote{Specifically, the CSI model in \cite{X.Chen_2017} replaced the time-varying terms $\sin^2 \theta_{n,k}$ and $\cos^2 \theta_{n,k}$ of the channel model in \cite{Jindal06,Yoo_limited,Liu2016,Q.Yang_2017}  (see (\ref{eq:CSIT_model}) in this paper) by the constant terms $2^{-\frac{B_{n,k}}{M-1}}$ and $1- 2^{-\frac{B_{n,k}}{M-1}}$ respectively.}, which is impractical for FDD systems. As we will show with the analysis, the ergodic rate performance with this model effectively underestimates the actual performance. Another minor weak point with the analytical performance result of \cite{X.Chen_2017} is it can not apply for some special power allocation scenarios.
%In addition, in contrast to the constraint on the system parameters in \cite{Q.Yang_2017} that the number of user clusters $N$ is no larger than the number of transmit antennas $M$, the constraint in \cite{X.Chen_2017} required $N$ and the number of the scheduled users $K$ of each cluster in each channel use satisfying $M > (N-1)K $. Thus, this scheme can be identified between the two categories of multi-antenna NOMA schemes illustrated above.

Moreover, we find it an essentially very hard problem to obtain the \emph{exact} analytical performance of any multi-antenna NOMA scheme based on the linear ZFBF at the transmitter for general system settings (e.g. \cite{Q.Yang_2017,X.Chen_2017}). The analytical performance of both schemes in \cite{Q.Yang_2017,X.Chen_2017} were actually obtained by ignoring the statistical dependence between the useful signal term and the interference terms in each user's received signal. However, it is very difficult to analyze the effect of the aforementioned ignorance of the statistical dependence on the resulting analytical performance.
%very difficult to see the aforementioned ignorance of the statistical dependence leads to the overestimation or underestimation of performance.
%As a result, the numerical results in \cite{Q.Yang_2017} showed that the aforementioned ignorance of the statistical dependence led to the analytical outage worse than that of the actual performance, especially in high power regime.
For the optimization of feedback bits allocation, assuming the numbers of the feedback bits of all users were equal, \cite{Q.Yang_2017} optimized the feedback by maximizing the net throughput of the system.
In contrast to maximizing the ergodic sum rate (ESR) of system, \cite{X.Chen_2017} proposed an intuitive method that aimed to minimize the average sum power of inter-cluster interference.
%the method proposed in \cite{X.Chen_2017} aimed to minimize the average sum power of inter-cluster interference, which was an intuitive method and sub-optimal in the system ESR.

%In contrast to directly maximizing the ESR of all users, the method proposed in \cite{X.Chen_2017} for the optimization of feedback bits allocation to each user %with the constraint of total a total bandwidth of all users' feedback channels
%aimed to minimize the average sum power of all interference power.

%\subsection{Contributions of This Paper}

In this paper, we provide a comprehensive study on the performance of multi-antenna NOMA and the joint power and feedback bits allocation optimization under the considered system. We reveal several important differences with respect to the solutions designed previously for multi-antenna NOMA downlink communication systems operating in FDD mode with limited CSI feedback, as well as provide practical implementation and engineering guidelines. Our main contributions are summarized as follows:
\begin{itemize}
  \item
  Motivated by the limitations with the previous works, we reconsider the design of the multi-antenna NOMA systems in FDD mode including user clustering, joint optimization of power and feedback bits allocation, limited CSI feedback, transmit BF and SIC at receivers.
%  In the previous works, random or fixed user clustering was usually employed for multi-antenna NOMA systems, and clustering method based on quantized CSIT was also usually employed for system with limited feedback.
  In this paper, we propose a low-complexity dynamic user clustering based on the large-scale fading of each user only, which can reduce the complexity of power allocation optimization and at the same time takes into consideration the user fairness.
  \item
  For a given power and feedback bits allocation, we provide a mathematically strict statistical analysis on the ergodic rate of each user taking into consideration the statistical dependence between the useful signal term and the interference terms in each user's received signal, which has never been done before. Our analysis based on the well recognized limited feedback framework results in a closed-form lower bound on the ergodic rate of each user without assuming any extreme for system parameters, a result that, to the best of the authors’ knowledge, has not been previously presented in the literatures.
  \item
  For a given power allocation result to all users, we obtain a closed-form expression for the optimal feedback bits allocation (without integer constraint) to minimize an upper bound on the ESR loss due to quantized CSI feedback. Then, we devise a low-complexity dynamic programming algorithm to find the optimal practical bits allocation solution.
  \item
  For the transmit power allocation, we first follow the intuitive power allocation scheme in \cite{X.Chen_2017} to employ equal power allocation among the users within each cluster. Then, we observe that the ergodic rate of the users other than the nearest user (to the BS) in each cluster is relatively much smaller than that of the nearest user, and the ESR of these users within each cluster is also very limited compared with that of the nearest user. Thus, we seek for the solution by maximizing the ESR of the nearest users of all clusters, which is generally a non-convex optimization problem.
  Then, based on the closed-form solution of the feedback bits allocation obtained above, we propose a closed-form sub-optimal solution to the power allocation among the different clusters. Finally, a low-complexity scheme of joint optimization of the transmit power and feedback bits allocation is obtained.
  \item
  Through asymptotic analysis of the joint power and feedback bits allocation solution, we observe several key insights. Specially, for the scenario with high CSI quantization accuracy (or equivalently large total bandwidth of feedback channels $B$) and finite total transmit power $P$, our power allocation solution tends to be the well known water-filling type solution for the system with perfect CSIT. For the scenario in the high power region with finite $B$, our power allocation solution tends to be equal power allocation among all users. Moreover, in the high power region with large enough $B$, the optimal feedback bits allocated to the nearest user in each cluster is scaled approximately linearly increasing as $\log_2 (P)$, whilst the corresponding one for the user other than the nearest user is scaled linearly approximately decreasing as $\log_2 (P)$. For the scenario with finite $B$ and large enough $P$, the number of feedback bits allocated to the user other than the nearest user in each cluster reduces to zero, whilst the one corresponding to the nearest user of each cluster approximately converges to a finite number that can be given in closed-form without integer constraint.
\end{itemize}

\emph{Notations}: $\mathcal{C}$ and $\mathcal{N}$ denote the sets of complex numbers and natural numbers respectively. $\mathbb{E}_{X}\{ \cdot \}$ represents expectation with respect to random variable $X$. $ \| \cdot\|$ denotes the $L_2$-norm of
a vector, $| \cdot|$ denotes the absolute value of a scalar. %$\lceil  x \rceil$ denotes the smallest integer larger than $x$.

\section{System and Signal Models}
We consider downlink communications in a single-cell cellular network operating in FDD mode for delay-tolerant traffics, where the coding block of each user's information can be sufficiently long to cover multiple channel coherent time periods such that the ergodic rate can be used as the performance metric. This scenario has been widely considered in the literature of NOMA \cite{Ding_NOMA2014,Ding.UserPairing2016,ZhiguoDing2016,X.Chen_2017,L.Zhang_JSAC_2017}.
There is a BS equipped with $M$ antennas and a large number of users with single-antenna each. The BS simultaneously broadcasts information to multiple scheduled users by BF at each channel use. Power-domain NOMA technique is concurrently employed to further improve the spectral efficiency, where all $NK$ scheduled users
are grouped into $N$ ($ N\geq 1$) clusters with $K$ ($ K\geq 2$)  users
within each cluster being simultaneously supported by NOMA\footnote{For the convenience of presentation, we have assumed the number of users of each cluster to be equal to $K$. It is easy to see, with very simple generalization, our NOMA framework can simultaneously support arbitrary number of users (with some constraint by $M$), and can be applicable to the scenarios with different number of users in each cluster.}\cite{B.Kim_downlink_NOMA}.
%\footnote{For the convenience of presentations, we have assumed the number of users of each cluster to be equal to $K$ and the total number of users can be divided by $K$. It is easy to see our NOMA framework can be applicable to the scenarios with different number of users in each cluster and can simultaneously support arbitrary number of users with very simple generalization.}
Our user clustering method will be introduced in the following. We denote the $M$-dimensional channel vector from the BS to the $k$-th user in the $n$-th cluster (denoted as user $(n, k)$ afterward) as $d_{n,k}^{-\frac{\alpha}{2}} \mathbf{h}_{n,k}$, where $\mathbf{h}_{n,k} \in \mathcal{C}^{1\times M}$ and $d_{n,k}$ are respectively the fast-fading channel and the distance from the BS to user $(n, k)$. We focus on the MU communications on a target time-frequency resource block (RB). Thus, all channels are assumed to undergo block flat-fading with a path-loss governed by the exponential coefficient $\alpha$.

\subsection{Channel and Feedback Models}

Since we consider the system in FDD mode, to focus on the effect of limited feedback, we follow many previous works (e.g. \cite{Ding_NOMA2016,Liu2016,Q.Yang_2017}) to assume each user can perfectly estimate the downlink CSI, and then feeds back the quantized CSI to the BS through an error-free but limited-rate feedback channel. The total bandwidth of the feedback channels of all users is constrained to be $B$ bits.

As \cite{X.Chen_2017}, the CSI obtained by user $(n,k)$ is quantized using a codebook $ \mathbb{C}_{n,k}$, which consists of $2^{B_{n,k}}$ codewords. Here $B_{n,k}$ is the number of feedback bits allocated to user $(n,k)$.
%$\mathcal{C}_{n,k} = \{\mathbf{c}_{n,k}{(1)}, \dots, \mathbf{c}_{n,k}{({B_{n,k}})}\}$, where $N_{n,k} = 2^{B_{n,k}}$ the size of
%the codebook.
Since the optimal CSI quantization strategy is unknown in general and is out of the scope of this work, we employ the same random quantization codebooks as those in \cite{Jindal06,Yoo_limited}. The rule for user $(n, k)$ to quantize the channel direction information (CDI) (i.e., $\tilde{\mathbf{h}}_{n,k} = \frac{\mathbf{h}_{n,k}}{||\mathbf{h}_{n,k}||}$) is given by $\hat{\mathbf{h}}_{n,k} = \arg \max \limits_{ \mathbf{c} \in \mathcal{C}_{n,k}  } \left| \tilde{\mathbf{h}}_{n,k} \mathbf{c}^H\right|$. Then, $\tilde{\mathbf{h}}_{n,k}$ can be decomposed as\cite{Jindal06}
\begin{eqnarray}\label{eq:CSIT_model}
\tilde{\mathbf{h}}_{n,k} = \cos\theta_{n,k} \hat{\mathbf{h}}_{n,k} + \sin\theta_{n,k} \tilde{\mathbf{e}}_{n,k},
\end{eqnarray}
where $\theta_{n,k} = \angle(\tilde{\mathbf{h}}_{n,k}, \hat{\mathbf{h}}_{n,k})$ and $\tilde{\mathbf{e}}_{n,k}$ is the normalized quantization error vector that is isotropically distributed in the nullspace of $\hat{\mathbf{h}}_{n,k}$. Note that, due to very limited bandwidth of feedback channels, only the index of quantized CDI in the codebook of each user is fed back.
For the transmitter to obtain the large-scale fading (or $d_{n,k}^{-\alpha} $), each user measures ``reference signal receiving power'' through downlink cell-specific reference signals (for long term evolution (LTE) systems) or CSI-reference signals (for LTE-advanced systems) \cite{TS36.211} and then feeds back mean received signal power to the transmitter.
Since the required frequency to feed back this information is much lower than that of fast fading (i.e., quantized CDI), we ignore its cost and follow many previous related papers to assume path loss of each user is known by the transmitter.

\vspace{-3mm}
\subsection{User Clustering}

Most of the existing user clustering methods for multi-antenna NOMA systems are according to fast-fading channels.
%especially the channel direction information (CDI)\cite{Yoo_limited} of each user, i.e., $\widetilde{\mathbf{h}} = \frac{\mathbf{h}}{\| \mathbf{h}\|}$.
For instance, the semi-orthogonal user clustering schemes based on CDI were proposed in \cite{Liu2016}. However, this kind of methods carries out exhaustive search resulting in high implementation complexity, and additionally requires feedback of accurate enough channel quality information (e.g., channel gain or SINR) for user clustering and determining the order of SIC within each cluster, which costs much more bandwidth of feedback channels for them to work correctly.
%Certainly, user clustering methods in NOMA systems can be designed from other perspectives, such as maximizing instantaneous (or average) output signal-to-interference-plus-noise ratio (SINR) or sum rate of each user cluster.
In addition, user clustering also depends on the power allocation to each user in general, which makes the kind of the methods based on instantaneous CSI even more complicated.
%In contrast, \cite{X.Chen_2017} proposed a low-complexity method which actually randomly formed the user clusters based on large-scale fading of each user.
In spite of all existing methods, considering delay-tolerant traffics and only very limited CSIT of each user\footnote{Recall that only quantized CDI and the large-scale fading (path-loss) are available at the transmitter. Thus, user clustering method based on full CSIT of fast-fading channel at the transmitter is impossible.},
%\footnote{As many previous works, we consider the practical scenario of FDD systems where only quantized channel direction information (CDI)\cite{Yoo_limited} (i.e., $\widetilde{\mathbf{h}} = {\mathbf{h}}/{\| \mathbf{h}\|}$ where $\mathbf{h}$ is a channel vector) of each user can be obtained by the transmitter through feedback channel with limited bandwidth and the large-scale fading
%(path-loss) can be obtained at the transmitter through reverse-link reference signal measurements. Then, the user clustering method based on full CSIT of fast-fading channel employed by the transmitter is impossible.}
%The limited CSI feedback will be described in details in the next subsection.}
we propose a low-complexity clustering strategy based on the large-scale fading of each user only. Considering the user fairness, our scheme first randomly selects $N K$ users each time, and then forms $N$ user clusters with $K$ users within each following the criteria that
%\footnote{Specifically, all users are ordered as $\frac{d_{1, 1}^{-\alpha}  }{\sigma_{1,1}^2} \leq \frac{d_{2, 1}^{-\alpha} }{\sigma_{2,1}^2}  \leq  \cdots \leq  \frac{d_{N, 1}^{-\alpha}  }{\sigma_{N,1}^2} \leq  \frac{d_{1, 2}^{-\alpha}  }{\sigma_{1,2}^2}  \leq  \cdots \frac{d_{N, 2}^{-\alpha}  }{\sigma_{N,2}^2}  \leq \cdots  \leq \frac{d_{1, K}^{-\alpha}  }{\sigma_{1, K}^2} \leq  \cdots \leq \frac{d_{N, K}^{-\alpha}  }{\sigma_{N , K}^2} $.}
%We propose to form the user clusters and sort the users in the same cluster following the criteria that
\begin{eqnarray}\label{eq:cluster_criteria1}
\frac{d_{n, i}^{-\alpha}  }{\sigma_{n,i}^2} & \leq & \frac{d_{m, i}^{-\alpha}}{\sigma_{m,i}^2} ~ ~\text{for} ~ 1 \leq  m < n\leq N , \; \forall 1\leq i \leq K ; \\
\label{eq:cluster_criteria2}
\frac{d_{n, i}^{-\alpha}  }{\sigma_{n,i}^2} & \leq &  \frac{d_{n, j}^{-\alpha}  }{\sigma_{n,j}^2}~~ \text{for} ~  1\leq j \leq i \leq K, \; \forall   1\leq n \leq N,
\end{eqnarray}
%(a) $\frac{d_{n, i}^{-\alpha}  }{\sigma_{n,i}^2} \leq \frac{d_{n, i}^{-\alpha}}{\sigma_{n,i}^2}  $ for$1 \leq  m < n\leq N$ and $\forall 1\leq i \leq K$; (b) $\frac{d_{n, i}^{-\alpha}  }{\sigma_{n,i}^2} \leq \frac{d_{n, j}^{-\alpha}  }{\sigma_{n,j}^2}$ for $  1\leq i \leq j  \leq N$ and $\forall 1\leq n \leq N$,
%
%i.e., $\frac{d_{1, 1}^{-\alpha}  }{\sigma_{1,1}^2} \leq \frac{d_{2, 1}^{-\alpha} }{\sigma_{2,1}^2}  \leq  \cdots \leq  \frac{d_{N, 1}^{-\alpha}  }{\sigma_{N,1}^2} \leq  \frac{d_{1, 2}^{-\alpha}  }{\sigma_{1,2}^2}  \leq  \cdots \frac{d_{N, 2}^{-\alpha}  }{\sigma_{N,2}^2}  \leq \cdots  \leq \frac{d_{1, K}^{-\alpha}  }{\sigma_{1, K}^2} \leq  \cdots \leq \frac{d_{N, K}^{-\alpha}  }{\sigma_{N , K}^2} $,
where $\sigma_{n,k}^2$ is the variance of the additive white Gaussian noise (AWGN) at user $(n, k)$.
It will be shown in Section \ref{subsec:power_allocation} after Remark \ref{remark:last} that, the criteria working with our proposed joint power and feedback bits allocations maximize the system ESR.

\subsection{Zero-Forcing Beamforming and Successive Interference Cancellation}

To mitigate MU interference, we employ the widely used simple linear ZFBF at the BS\cite{Liu2016,Q.Yang_2017,Jindal06,Yoo_limited,X.Chen_2017}, where the BF vector $\mathbf{w}_n$ of cluster $n$ satisfies $\hat{\mathbf{h}}_{m,k} \mathbf{w}_n = 0$ $\forall m \neq n$ and $ \forall k$. Following the method in \cite{X.Chen_2017},
for cluster $n$ we construct a complementary matrix $\bar{\mathbf{H}}_n = \Big[\hat{\mathbf{h}}_{1,1}^H, \dots, \hat{\mathbf{h}}_{1,K}^H , \dots, \hat{\mathbf{h}}_{n-1,K}^H, \hat{\mathbf{h}}_{n+1,1}^H, \dots, \hat{\mathbf{h}}_{N-1,K}^H, \hat{\mathbf{h}}_{N,1}^H ,\\\dots, \hat{\mathbf{h}}_{N,K}^H \Big]^H \in \mathcal{C}^{(N-1)K \times M}$.
Let the singular-value decomposition (SVD) of $\bar{\mathbf{H}}_n$ be $\bar{\mathbf{H}}_n = \bar{\mathbf{U}}_n  \bar{\mathbf{\Sigma}}_n [\bar{\mathbf{V}}_n^{(1)}, \bar{\mathbf{V}}_n^{(0)}]$, where the columns of matrix $\bar{\mathbf{V}}_n^{(0)} \in \mathcal{C}^{M \times (M - \bar{r}_n)}$ are the right-singular vectors corresponding to the zero singular-values. Here $\bar{r}_n \triangleq \text{rank} ( \bar{\mathbf{H}}_n) = (N - 1)K$ with probability 1 when the channel vector of each user is continuously distributed.
Then, $\mathbf{w}_n$ can be obtained as $\mathbf{w}_n = \bar{\mathbf{V}}_n^{(0)} \mathbf{p}_n$\cite{X.Chen_2017}, where $\mathbf{p}_n$ is uniformly distributed on the surface of the $(M - \bar{r}_n)$-dimensional unit sphere. It is easy to see the constraint $M > (N-1)K$ needs to be satisfied for ZFBF to work properly.
Then, with the decomposition of channel in (\ref{eq:CSIT_model}), the received signal at user $(n, k)$ can be written as
\vspace{-1mm}
\begin{eqnarray}
\hspace{-9mm} y_{n,k} &=& d_{n,k}^{-\frac{\alpha}{2}} \mathbf{h}_{n,k} \sum_{i=1}^N \mathbf{w}_{i} s_{i} + n_{n,k} \nonumber\\
\hspace{-9mm} &=& d_{n,k}^{-\frac{\alpha}{2}} ||\mathbf{h}_{n,k}|| \sin\theta_{n,k} \sum_{i=1,i\neq n}^N \tilde{\mathbf{e}}_{n,k} \mathbf{w}_{i} s_{i} \nonumber\\
\hspace{-9mm}&&+~ d_{n,k}^{-\frac{\alpha}{2}} \mathbf{h}_{n,k} \mathbf{w}_{n} s_{n}+ n_{n,k}, \nonumber
\end{eqnarray}
where $s_i = \sum_{j=1}^{K} \sqrt{P_{i,j}} s_{i,j}$ is the superposition coded signals for all users in cluster $i$
with $s_{i,j}$ and $P_{i,j}$ being the signal of user $(i,j)$ and the corresponding power, and $n_{n,k}$ being the AWGN.
In general, $P_{i,j}$ should be properly allocated to enable successful SIC and ensure certain level of ESR performance of each user and also the fairness among the users within a cluster.

Although ZFBF can cancel partial interference, there still exists inter-cluster and intra-cluster interference due to limited CSI feedback. According to the principle of NOMA, the users other than user $(n , K)$ in each cluster conduct SIC to recover the information.
%We assume that the BS can know the MUs' effective gains by the channel quality indicator (CQI), then, the MUs in the same cluster can be sorted as follow
%\begin{eqnarray}
%\frac{d_{n,1}^{-\alpha}}{\sigma_{n,1}^2} |\mathbf{h}_{n,1} \mathbf{w}_{n}|^2 \geq \cdots \geq \frac{d_{n,k}^{-\alpha}}{\sigma_{n,k}^2} | \mathbf{h}_{n,k} \mathbf{w}_{n}|^2 \geq \cdots \geq \frac{d_{n,K}^{-\alpha}}{\sigma_{n,K}^2} |\mathbf{h}_{n,K} \mathbf{w}_{n}|^2.
%\end{eqnarray}
Recall that it is assumed the coding block of the considered systems can be sufficiently long. Thus, according to information theory, if the transmission rate of user $(n,k)$ does not exceed the ergodic rate of user $(n, k)$ supported by the multiuser fading channels, the decoding error with SIC at user $(n, k)$ can be arbitrarily small as the coding block length of each user goes large.
Moreover, for the same reason as noted in \cite{X.Chen_2017,Q.Yang_2017}, when the power is properly allocated among the signals of the users within each cluster, user $(n, k)$ can always successfully decode the user $(n, j)$'s signal for $\forall j > k$ with the user ordering given by (\ref{eq:cluster_criteria2}), if user $(n, k)$ can decode its own signal\footnote{It is easy to see that our proposed method in Subsection \ref{subsec:power_allocation} that allocates equal power to the users within each cluster satisfies this requirement.}.
As a result, before decoding its own signal, user $(n, k)$ can cancel the interference from user $(n, j)$ for $ j>k $ in the received signal.
Thus, we can assume perfect SIC can be performed at each user as \cite{X.Chen_2017,Q.Yang_2017}. Then, after SIC, the SINR at user $(n, k)$ for $k=1$ and $ k >1$ are given respectively by (\ref{eq:gamma_nk_eq1}) and (\ref{eq:gamma_nk_gt1}) at the top of next page.
\begin{figure*}
\begin{eqnarray}\label{eq:gamma_nk_eq1}
\gamma_{n,1} &=& \frac{d_{n,1}^{-\alpha} |\mathbf{h}_{n,1} \mathbf{w}_{n}|^2 P_{n,1}}
{\underbrace{d_{n,1}^{-\alpha} \sin^2\theta_{n,1} ||\mathbf{h}_{n,1}||^2 \sum_{i=1,i \neq n}^{N} |\tilde{\mathbf{e}}_{n,1} \mathbf{w}_{i}|^2 \sum_{l=1}^{K} P_{i,l}}_{\text{Inter-cluster interference}} + \sigma_{n,1}^2},\\
\label{eq:gamma_nk_gt1}
\gamma_{n,k} &=& \frac{d_{n,k}^{-\alpha} |\mathbf{h}_{n,k} \mathbf{w}_{n}|^2 P_{n,k}}
{\underbrace{d_{n,k}^{-\alpha} |\mathbf{h}_{n,k} \mathbf{w}_{n}|^2 \sum_{j=1}^{k-1} P_{n,j}}_{\text{Intra-cluster interference}} +
\underbrace{d_{n,k}^{-\alpha} \sin^2\theta_{n,k} ||\mathbf{h}_{n,k}||^2 \sum_{i=1,i \neq n}^{N} |\tilde{\mathbf{e}}_{n,k} \mathbf{w}_{i}|^2 \sum_{l=1}^{K} P_{i,l}}_{\text{Inter-cluster interference}} + \sigma_{n,k}^2} .
\end{eqnarray}
\vspace{-6mm}
\hrulefill
\end{figure*}
The SIC may not be perfect in practice due to limited capability of users. However, the study of the impact of imperfect SIC is out of the scope of this paper and will be considered in the future work.
\begin{remark}
We note that linear ZFBF is widely recognized as a low-complexity processing. The main computation complexity of ZFBF lies in the SVD of matrix $\bar{\mathbf{H}}_n $ for $n=1, 2, \cdots, N$, which requires the total number of floating point operations per second (FLOPS) of $~32 N \Big[  M (N-1)^2K^2 + 2(N -1)^3   K^3 \Big]$\cite{matrix_com}. Acquiring knowledge of complete CSI or even partial fast fading at transmitter increases system complexity. Since the path loss of each user remains constant for a relatively long time, our user clustering and ordering method only updates at much lower frequency than the update of BF vectors which is carried out with the period of channel coherent time. Thus, our method is of low-complexity.
\end{remark}
The power allocation and feedback bits allocation have significant impact on the performance of the NOMA systems. Before optimizing the power and feedback bits allocation, we first analyze the ergodic rate of the considered system.

\section{Performance Analysis}\label{sec:actual_performance}
In this section, we focus on obtaining the analytical ergodic rate of each user given power and bits allocation and without assuming any extreme for system parameters. All fast-fading channels are assumed to follow circularly symmetric complex Gaussian distribution with zero mean and unit variance, i.e., Rayleigh fading.
We note that, since the user ordering given by (\ref{eq:cluster_criteria1}) and (\ref{eq:cluster_criteria2}) does not depend on the fast fading,
the distribution of each fast-fading term of different scheduled users is the same. The ergodic rate of user $(n, k)$ is given by $R_{n,k} = \Bbb{E}[ \log_2(1 + \gamma_{n,k})]$ with $\gamma_{n,k}$ given by (\ref{eq:gamma_nk_eq1}) or (\ref{eq:gamma_nk_gt1}) at the top of next page.
%\begin{eqnarray}\label{eq:R_nk_raw}
%R_{n,k} = \Bbb{E}[ \log_2(1 + \gamma_{n,k})],
%\end{eqnarray}
%where $\gamma_{n,k}$ is given by (\ref{eq:gamma_nk_eq1}) and (\ref{eq:gamma_nk_gt1}).
First, we note that, due to the statistical dependence between useful signal term and interference terms at each user, it is very difficult if not impossible to obtain the exact analytical result of $R_{n,k}$. Thus, we will turn to obtaining an analytical result as accurate as possible. Specifically, we can obtain a lower bound of $R_{n,k}$ in the following theorem.
\begin{theorem}\label{theorem:R_nk_LB}
$R_{n,k}$ $\forall n, k$ can be lower-bounded as $R_{n,k} \geq R_{n,k }^{LB1}$ with $R_{n,k }^{LB1}$ given by
\begin{eqnarray}\label{eq:R_nk_LB1_1}
\hspace{-0.8cm}&& R_{n, 1}^{LB1} = \Theta_{n,k}(\alpha_{n,1} P_{n,1}, S_{n,1}^{(3)}(M-1)^{-1} \delta_{n,1})  \nonumber\\
\hspace{-0.8cm}&&\hspace{0.4cm}- \log_2(e) \exp \left( \frac{M-1}{S_{n,1}^{(3)} \delta_{n,1}} \right) \sum_{q=1}^{M-1} E_q \left( \frac{M-1}{S_{n,1}^{(3)} \delta_{n,1}} \right), ~~~~~\\
%\end{eqnarray}
%\begin{eqnarray}
\label{eq:R_nk_LB1_k}
\hspace{-0.5cm}  &&R_{n, k}^{LB1} = \Theta_{n,k} \left( S_{n,k}^{(1)}, S_{n,k}^{(3)}(M-1)^{-1} \delta_{n,k} \right) \nonumber\\
\hspace{-0.8cm}&&\hspace{0.4cm}- \Theta_{n,k} \left( S_{n,k}^{(2)}, S_{n,k}^{(3)}(M-1)^{-1} \delta_{n,k} \right), \text{for} ~k>1,~
\end{eqnarray}
%\begin{eqnarray}\label{eq:R_nk_LB1}
%R_{n,k }^{LB1} =
%\begin{cases}
%\Theta_{n,k}(\alpha_{n,1} P_{n,1}, S_{n,1}^{(3)}(M-1)^{-1} \delta_{n,1}) \\
%\hspace{1cm} - \log_2(e) \exp \left( \frac{M-1}{S_{n,1}^{(3)} \delta_{n,1}} \right) \sum_{q=1}^{M-1} E_q \left( \frac{M-1}{S_{n,1}^{(3)} \delta_{n,1}} \right), &k=1 \\
%\Theta_{n,k} \left( S_{n,k}^{(1)}, S_{n,k}^{(3)}(M-1)^{-1} \delta_{n,k} \right)  \\
%\hspace{1cm} - \Theta_{n,k} \left( S_{n,k}^{(2)}, S_{n,k}^{(3)}(M-1)^{-1} \delta_{n,k} \right), &k>1
%\end{cases},
%\end{eqnarray}
where $S_{n,k}^{(1)} = \frac{d_{n,k}^{-\alpha}}{\sigma_{n,k}^2} \sum_{j=1}^k P_{n,j}$, $S_{n,k}^{(2)} = \frac{d_{n,k}^{-\alpha}}{\sigma_{n,k}^2} \sum_{j=1}^{k-1} P_{n,j}$, $S_{n,k}^{(3)} = \frac{d_{n,k}^{-\alpha}}{\sigma_{n,k}^2} \sum_{i=1,i \neq n}^{N} \sum_{j=1}^{K} P_{i,j}$, $\delta_{n,k} = 2^{-\frac{B_{n,k}}{M-1}}$,
and $E_q(x) = \int_{1}^{\infty} \frac{e^{-xt}}{t^q} dt$ is generalized exponential integral. $\Theta_{n,k} (a,b)$ is given by
%\vspace{-5mm}
\begin{eqnarray}
\hspace{-7mm}&&\Theta_{n,k} (a,b) =  \log_2(e) \Bigg[ \frac{(-1)^M(M-1)}{b^{M-1}} I_1(a) \nonumber\\
\hspace{-7mm}&&+ \sum_{p=1}^{M-1} \sum_{q=1}^{M-p} \frac{ (-1)^{p+q+1} (M-1) }{ (q-1)! b^p } I_2(a,b,p,q) +
\label{eq:res_Theta_nk} \\
\hspace{-7mm}&& \sum_{p=1}^{M-1} \sum_{q=2}^{M-p} \sum_{s=0}^{q-2} \frac{ (-1)^{p+s-1} (q-s-2)! (M-1) }{ (q-1)! b^{p} } I_3(a,b,p,s)
\Bigg] \nonumber
\end{eqnarray}
with $I_1(a) $, $I_2(a,b,p,q)$ and $I_3(a,b,p,s) $ given respectively as
\vspace{-1mm}
\begin{eqnarray}
&& \hspace{-0.7cm}  I_1(a) = \sum_{t=0}^{M-2} {M-2 \choose t} \frac{(-1)^t}{a^{t+1}} \tilde{\Psi}\left(-M-1-t, \frac{1}{a} \right), \nonumber \\
&& \hspace{-0.7cm}  I_2(a,b,p,q) = \sum_{r=0}^{M-2} \sum_{t=0}^{p-1+r} {M-2 \choose r} {p-1+r \choose t} \nonumber \\
&&\hspace{0.7cm}\times \frac{(-1)^{p-1-t} b^{p-1+r-t}}{a^{r+1}} \Psi \left(q-4-t, \frac{1}{a+b}, \frac{1}{b}\right), \nonumber \\
&&  \hspace{-0.7cm}  I_3(a,b,p,s)  = \sum_{t=0}^{M-2} {M-2 \choose t} \frac{ (-1)^{t} a^{p-1} }{ b^{s-1} (p+t) } \nonumber \\
&&\hspace{0.7cm}\times ~{_2} F_1 \left( s-1, p+t; p+t+1; -\frac{a}{b} \right), \nonumber
\end{eqnarray}
where $_2F_1(a,b;c;z) $ is a Gauss-hypergeometric function \cite[9.14.2]{Gradshteyn2007} and $\Psi (n, u, v)$ is given by
\begin{align}\label{eq:res_Psi_1}
&\Psi(n, u, v) \triangleq \nonumber \\
&\begin{cases}
\sum_{k=0}^n \frac{(-1)^k n! \left[ v^{n-k} e^v E_i(-v) - u^{n-k} e^u E_i(-u) \right]}{(n-k)!}\\
\hspace{4mm} - \sum_{k=0}^{n-1} \frac{(-1)^k n! (v^{n-k} - u^{n-k}) }{(n-k)(n-k)!}  - (-1)^n n! \ln \left( \frac{v}{u} \right), n>0 \\
e^v E_i(-v) - e^u E_i(-u) - \ln \left( \frac{v}{u} \right), ~~~~~~~~~~~~~~~~~~~~n=0\\
E_i(-v)E_i(v) - E_i(-u)E_i(u) + \frac{E_i^2(-u) - E_i^2(-v)}{2} \\
\hspace{4mm} + \frac{1}{2}\sum_{m=0}^{\infty} \sum_{l=0}^{2m} \frac{2 (v^l e^{-v} - u^l e^{-u})}{(2m+1)^2\,l!}, ~~~~~~~~~~~~~~~~n=-1 \\
\sum_{k=1}^{-n-1} \frac{(-n-k-1)!}{(-n-1)!} \left[ u^{n+k} e^u E_i(-u) - v^{n+k} e^v E_i(-v) \right] \\
\hspace{4mm}+ \frac{E_i(-v)E_i(v) - E_i(-u)E_i(u)}{(-n-1)!} + \\
\hspace{4mm}\sum_{k=1}^{-n-1} \frac{ (-n-k-1)!(v^{n+k} - u^{n+k}) }{ (-n-1)!(n+k) } + \frac{E_i^2(-u) - E_i^2(-v)}{2(-n-1)!} + \\
\hspace{4mm}\frac{1}{(-n-1)!} \sum_{m=0}^{\infty} \sum_{l=0}^{2m} \frac{2 (v^l e^{-v} - u^l e^{-u})}{(2m+1)^2\,l!}, ~~~~~~~~~~n\leq -2.
\end{cases}
\end{align}
$E_i(x) = -\int_{-x}^{\infty} \frac{e^{-t}}{t} dt$ is the exponential integral. For $n < 0$ with $v \to +\infty$, $\Psi(n, u, v) $ becomes
\begin{eqnarray}\label{eq:res_Psi_2}
\hspace{-7mm}&&\tilde{\Psi}(n, u) \triangleq {\Psi}(n, u, v) \big|_{v \to +\infty} \nonumber\\
\hspace{-7mm}&&=\frac{ e^u E_i(-u) }{(-n-1)!} \sum_{k=1}^{-n-1} (-n-k-1)! u^{n+k} - \frac{ E_i(-u) E_i(u) }{(-n-1)!} \nonumber\\
\hspace{-7mm}&&\hspace{4mm} -\sum_{k=1}^{-n-1} \frac{ (-n-k-1)!~u^{n+k} }{ (-n-1)! (n+k) }
+ \frac{ E_i^2(-u) }{2 (-n-1)!} \nonumber\\
\hspace{-7mm}&&\hspace{4mm}- \frac{ 1 }{(-n-1)!} \sum_{m=0}^{\infty} \sum_{l=0}^{2m} \frac{2 u^l e^{-u}}{(2m+1)^2\,l!}. ~~
\end{eqnarray}
\end{theorem}
\begin{proof}
See Appendix \ref{proof:theorem_R_nk_LB}.
\end{proof}
\begin{remark}
As it was also observed by the authors in \cite{Q.Yang_2017}, the exact analytical performance of the multi-antenna NOMA systems based on linear ZFBF with limited CSI feeback model here is very difficult to obtain. We believe the actual analytical result (if it could be obtained) is at least as the same complicated as the bound obtained in \emph{Theorem} \ref{theorem:R_nk_LB}, if not more complicated. Thus, it is natural that any accurate analytical lower or upper bound should be very complicated. In spite of this, the theoretical result obtained in Theorem \ref{theorem:R_nk_LB} is able to avoid the time-consuming computer simulations in evaluating the system performance.
%To the best of our knowledge, although there have been some analytical performance results about the multi-antenna NOMA system based on ZFBF and limited CSI feeback in \cite{Liu2016,Q.Yang_2017,X.Chen_2017},
\end{remark}
\begin{remark}\label{remark:2}
It can be observed from (\ref{eq:R_nk_LB_temp1}) in \emph{Appendix} \ref{proof:theorem_R_nk_LB}, the channel model with limited feedback used in \cite{X.Chen_2017} in fact equivalently further employs Jensen's inequality on the random term $\sin^2\theta_{n,k}$ in (\ref{eq:R_nk_LB_temp1}), thus results in lower ergodic rate performance than that of any practical system. This point will also be verified by the numerical result of Fig. \ref{fig:plot_compare_joint_optimization} in Section \ref{sec:numerical_results}.
\end{remark}

\section{Joint Optimization of Feedback Bits and Power Allocation}\label{sec:Bit_Power_Allocations}

The signal strength can be enhanced and the MU interference can be reduced by increasing the accuracy of CSI quantization. The transmit power allocation to each user also has a great impact on the ergodic rate performance. Moreover, it is obvious that the feedback bits allocation and the power allocation interact with each other.
Therefore, it is essential to consider the joint optimization of the feedback bits and power allocation among the MUs for performance enhancement with the constraints on the total transmit power and the sum bandwidth of feedback channels of all users. %In this section, we consider the problem of jointly optimizing feedback bits and power allocation among all users.

\subsection{Feedback Design Given Power Allocation Among All Users}\label{subsec:feedback_allocations}
In this subsection, we first concentrate on optimizing the ESR by feedback bits allocation among all users given the power allocation among all users. The power allocation optimization will be considered in the next subsection. \cite{X.Chen_2017} proposed a indirect method that minimized the average sum power of inter-cluster interference of all users, which can not guarantee maximizing the ESR. In contrast to \cite{X.Chen_2017}, we consider directly maximizing the ESR of all users or equivalently minimizing ESR loss caused by CSI quantization.
%The previous bits allocation scheme in \cite{X.Chen_2017} indirectly minimized the average sum power of inter-cluster interference, which can not guarantee maximizing the ESR of all users. Therefore, in contrast to the scheme in \cite{X.Chen_2017}, we consider directly maximizing the ESR of all users or equivalently minimizing ergodic sum rate loss caused by CSI quantization.
The ergodic rate loss of user $(n,k )$ due to CSI quantization is defined as $\Delta R_{n,k}  \triangleq R_{n,k}^{ideal} - R_{n,k} $, where $R_{n,k}^{ideal} =\Bbb{E} \left[ \log_2 \left( \frac{1 + S_{n,k}^{(1)} \left| \mathbf{h}_{n,k} \mathbf{w}_n \right|^2}{1 + S_{n,k}^{(2)} \left| \mathbf{h}_{n,k} \mathbf{w}_n \right|^2} \right) \right]$ is the ergodic rate with perfect CSI. Since the actual analytical result of $R_{n,k} $ is unknown and the obtained analytical lower bound in Section \ref{sec:actual_performance} is too sophisticated for further processing, in the following theorem we first develop an upper bound on the ergodic rate loss with a relatively simpler form.
\begin{theorem}\label{theorem:R_nk_Loss_Upper}
%The ergodic rate loss of user $(n,k )$ due to the quantized CSIT can be upper-bounded as $\Delta R_{n,k}  \triangleq R_{n,k}^{ideal} - R_{n,k} \leq \Delta R_{n,k}^{UB1}$, where $R_{n,k}^{ideal} =\Bbb{E} \left[ \log_2 \left( \frac{1 + S_{n,k}^{(1)} \left| \mathbf{h}_{n,k} \mathbf{w}_n \right|^2}{1 + S_{n,k}^{(2)} \left| \mathbf{h}_{n,k} \mathbf{w}_n \right|^2} \right) \right]$ is the ergodic rate with perfect CSI.
$\Delta R_{n,k} $ $\forall n, k$ can be upper-bounded as $\Delta R_{n,k} \leq \Delta R_{n,k}^{UB1}$, where
\begin{eqnarray}\label{eq:Delta_R_nk_UB}
\hspace{-8mm}&&\Delta R_{n,k}^{UB1}  =  \log_2 \left( 1 + S_{n,k}^{(2)} + \Gamma \left( \frac{2M-1}{M-1} \right) 2^{-\frac{B_{n,k}}{M-1}}  S_{n,k}^{(3)} \right) \nonumber\\
\hspace{-8mm}&&\hspace{16mm}- \log_2(e) e^{\frac{1}{S_{n,k}^{(2)}}} E_1 \left( \small{\frac{1}{S_{n,k}^{(2)}}} \right).
\end{eqnarray}
Here $ S_{n,k}^{(2)}$ and $S_{n,k}^{(3)} $ are the same as defined in \emph{Theorem} \ref{theorem:R_nk_LB}.
\end{theorem}
\begin{proof}
See Appendix \ref{proof:theorem_R_nk_Loss_Upper}.
\end{proof}

With \emph{Theorem} \ref{theorem:R_nk_Loss_Upper}, an upper bound on the ESR loss $\Delta R_{sum}$ can be obtained as
\begin{eqnarray}\label{eq:Delta_R_sum_UB}
\hspace{-12mm}&& \Delta R_{sum,UB} = \sum_{n=1}^N \sum_{k=1}^{K} \Delta R_{n,k}^{UB1} .
\end{eqnarray}
%\begin{eqnarray}\label{eq:Delta_R_sum_UB}
%&& \Delta R_{sum,UB} = \sum_{n=1}^N \sum_{k=1}^{K} \Delta R_{n,k}^{UB1}  \nonumber\\
%&&  =\sum_{n=1}^N \sum_{k=1}^{K} \left[ \log_2 \left( 1 + S_{n,k}^{(2)} + \Gamma \left( \frac{2M-1}{M-1} \right) 2^{-\frac{B_{n,k}}{M-1}}  S_{n,k}^{(3)} \right) - \log_2(e) e^{\frac{1}{S_{n,k}^{(2)}}} E_1 \left( \small{\frac{1}{S_{n,k}^{(2)}}} \right) \right]. ~~~~~
%\end{eqnarray}
Then, the feedback bits allocation optimization problem turns into the problem of minimizing $ \Delta R_{sum,UB} $.
Moreover, since the second term of (\ref {eq:Delta_R_nk_UB}) is irrelevant to $B_{n,k}$, the problem can be equivalently formulated as
\begin{eqnarray}\label{problem:bit_allocation_prod}
\min_{\{B_{n,k} \} } &&\hspace{-5mm} \prod_{n=1}^N \prod_{k=1}^{K} \left( 1 + S_{n,k}^{(2)} + \Gamma \left( \frac{2M-1}{M-1} \right) 2^{-\frac{B_{n,k}}{M-1}}  S_{n,k}^{(3)} \right)  \nonumber\\
\text{s.t.} &&\hspace{-5mm} \sum_{n=1}^{N} \sum_{k=1}^{K} B_{n,k} \leq B, \\
&&\hspace{-5mm} B_{n,k} \in \mathcal{N}, ~\forall n, k. \nonumber
\end{eqnarray}
Without non-negative integer constraint, the solution to (\ref{problem:bit_allocation_prod}) is obtained in the following theorem.
\begin{theorem}\label{theorem:B_without_integer}
The solution to the problem (\ref{problem:bit_allocation_prod}) without non-negative integer constraint is given by
\begin{eqnarray}\label{eq:B_nk}
\hspace{-8mm}B_{n,k}^{\star} &=& \frac{B}{NK} + (M-1)\log_2 \left( \frac{S_{n,k}^{(3)}}{1 + S_{n,k}^{(2)}} \right) \nonumber\\
\hspace{-8mm}&&+~ \frac{(M-1)}{NK} \sum_{p=1}^N \sum_{q=1}^K \log_2 \left( \frac{1 + S_{p,q}^{(2)}}{S_{p,q}^{(3)}} \right).
\end{eqnarray}
\end{theorem}
\begin{proof}
See Appendix \ref{proof:theorem_B_without_integer}.
\end{proof}
\begin{remark}
%We can show the superiority of our algorithm over that of \cite{X.Chen_2017} as follows.
By employing the arithmetic-geometric mean inequality, we have
%\begin{figure*}
%\begin{eqnarray}\label{eq:comparision}
%&& \prod_{n=1}^N \prod_{k=1}^{K} \left( 1 + S_{n,k}^{(2)} + \Gamma \left( \frac{2M-1}{M-1} \right) 2^{-\frac{B_{n,k}}{M-1}}  S_{n,k}^{(3)} \right) \leq \Bigg[ \frac{1}{NK} \sum_{n=1}^N \sum_{k=1}^{K} \Bigg( 1 + S_{n,k}^{(2)} + \Gamma \left( \frac{2M-1}{M-1} \right) 2^{-\frac{B_{n,k}}{M-1}}  S_{n,k}^{(3)} \Bigg) \Bigg]^{NK} \nonumber\\
%&& = \Bigg[ 1+ \frac{1}{NK} \Bigg( \sum_{n=1}^N \sum_{k=1}^{K} S_{n,k}^{(2)}   +  \Gamma \left( \frac{2M-1}{M-1} \right) \sum_{n=1}^N \sum_{k=1}^{K} 2^{-\frac{B_{n,k}}{M-1}}  S_{n,k}^{(3)} \Bigg) \Bigg]^{NK}.
%\end{eqnarray}
%\hrulefill
%\end{figure*}
\begin{eqnarray}\label{eq:comparision}
\hspace{-10mm}&&\prod_{n=1}^N \prod_{k=1}^{K} \left( 1 + S_{n,k}^{(2)} + \Gamma \left( \frac{2M-1}{M-1} \right) 2^{-\frac{B_{n,k}}{M-1}}  S_{n,k}^{(3)} \right) \nonumber\\
\hspace{-10mm}&&\leq \Bigg[ \frac{1}{NK} \sum_{n=1}^N \sum_{k=1}^{K} \Bigg( 1 + S_{n,k}^{(2)} + \Gamma \left( \frac{2M-1}{M-1} \right)\nonumber\\
\hspace{-10mm}&&\hspace{50mm}\times 2^{-\frac{B_{n,k}}{M-1}}  S_{n,k}^{(3)} \Bigg) \Bigg]^{NK} \nonumber\\
\hspace{-10mm}&& = \Bigg[ 1+ \frac{1}{NK} \Bigg( \sum_{n=1}^N \sum_{k=1}^{K} S_{n,k}^{(2)}   +  \Gamma \left( \frac{2M-1}{M-1} \right)\nonumber\\
\hspace{-10mm}&&\hspace{35mm}\times \sum_{n=1}^N \sum_{k=1}^{K} 2^{-\frac{B_{n,k}}{M-1}}  S_{n,k}^{(3)} \Bigg) \Bigg]^{NK}.
\end{eqnarray}
We notice the term $ \sum_{n=1}^N \sum_{k=1}^{K} 2^{-\frac{B_{n,k}}{M-1}}  S_{n,k}^{(3)} $ in the right hand side (RHS) of (\ref{eq:comparision}) is exactly the objective function of the optimal bits allocation problem formulated in \cite[(24)]{X.Chen_2017}. Since the other terms in the RHS of (\ref{eq:comparision}) is irrelevant to $\{B_{n,k}\}$, this illustrates minimizing the objective function of \cite{X.Chen_2017} is equivalent to minimizing an upper bound of the objective function of our problem. Thus, our algorithm is superior to that in \cite{X.Chen_2017} in general. We will also show the superiority of our algorithm by the numerical results in Section \ref{sec:numerical_results}.
\end{remark}

For some system parameter settings, (\ref{eq:B_nk}) may give negative values (particularly for user $(n,k)$ with $k>1$ in each cluster). In this case, our algorithm temporally allocates $0$ bit to these users.
%, i.e., the BS randomly determine the CDI of this user.
Then, the bits allocation algorithm given by \emph{Theorem} \ref{theorem:B_without_integer} is performed again for the remaining users. This algorithm operates in the recurrent manner until the bits allocation for all remaining users satisfies $B_{n,k} \geq 0$. Without causing confusion, we still denote the resulting bits allocation of user $(n, k)$ after the above operation as $B_{n,k}^{\star} $. Moreover, there is integer constraint on each $B_{n,k}$. Thus, based on the coarse bits allocation results obtained above, we seek for the practical solution by \emph{Algorithm} \ref{table:algorithm_table}.

\begin{algorithm}[t]
\caption{The Dynamic Programming Algorithm to Find Practical Feedback Bits Allocation} %算法的名字
\label{table:algorithm_table}
\hspace*{0.02in} {\bf Input:} %算法的输入， \hspace*{0.02in}用来控制位置，同时利用 \\ 进行换行
$\mathbf{b}_{rd}$, $\mathbf{r}_v(i)$, $B_{re}$\\
\hspace*{0.02in} {\bf Output:} %算法的结果输出
$\mathbf{b}_{int}$
\begin{algorithmic}[1]
\State $\mathbf{D}_p = \text{zeros}(NK+1, B_{re}+1)$;
\State $\mathbf{B}_{rc} = \text{zeros}(NK, NK+1, B_{re}+1)$;
\For{$i = 2$ to $NK+1$}
    \For{$j = 2$ to $B_{re}+1$}
        \If{$\mathbf{D}_p(i-1,j) > \mathbf{D}_p(i-1,j-1) + \mathbf{r}_v(i-1)$}
            \State $\mathbf{B}_{rc}(:,i,j) = \mathbf{B}_{rc}(:, i-1, j)$;
            \State $\mathbf{D}_p(i,j) = \mathbf{D}_p(i-1,j)$;
        \Else
            \State $\mathbf{B}_{rc}(:,i,j) = \mathbf{B}_{rc}(:, i-1, j-1)$;
            \State $\mathbf{B}_{rc}(i,i,j) = 1$;
            \State $\mathbf{D}_p(i,j) = \mathbf{D}_p(i-1,j-1) + \mathbf{r}_v(i-1)$;
        \EndIf
    \EndFor
\EndFor
\State \Return $\mathbf{b}_{int} = \mathbf{b}_{rd} + \mathbf{B}_{rc}(:, NK+1, B_{re}+1)$;
\end{algorithmic}
\end{algorithm}

First, let $\mathbf{b}_{int} \in \mathcal{N}^{NK \times 1}$ consist of the integer bits allocation to all users, where the $i$-th element $\mathbf{b}_{int}(i)$ is the number of bits allocated to user $(n',k')$ with the relationship between $i$ and $(n',k')$ given by $n'=\text{MOD}(i-1, N) + 1$ and $k' = \lceil i/N \rceil$. Here $\text{MOD}$ is the modular operation. Similarly, let $\mathbf{b}_{rd} \in \mathcal{N}^{NK \times 1}$ denote the bits allocation obtained as $\mathbf{b}_{rd} (i) = \lfloor\{B_{n',k'}^{\star} \} \rfloor$. In addition, let $\mathbf{r}_v  $ be $NK \times 1$ vector with the $i$-th element $\mathbf{r}_v(i) =  \Delta R_{n',k'}^{UB1} (\mathbf{b}_{rd}(i)) -  \Delta R_{n',k'}^{UB1} (\mathbf{b}_{rd}(i)+1)$ denoting the gain of the ergodic rate from allocating one extra bit to user $(n', k')$, where $\Delta R_{n',k'}^{UB1} $ is given by (\ref{eq:Delta_R_nk_UB}).
Let $B_{re} = B - \sum_{i=1}^{NK}\mathbf{b}_{rd} (i)$ denote the number of the remaining bits that need to be re-allocated.
Then, the problem to find the practical optimal solution can be formulated as
\begin{eqnarray}\label{problem:DP}
\max_{\mathbf{b}_{int}} && \left[\mathbf{b}_{int} - \mathbf{b}_{rd} \right]^T \mathbf{r}_v \\
\text{s.t.} && \sum_{i=1}^{NK} \left[\mathbf{b}_{int}(i) - \mathbf{b}_{rd}(i) \right] \leq B_{re}, \nonumber\\
&& \left[\mathbf{b}_{int}(i) - \mathbf{b}_{rd}(i) \right] \in \{0,1 \}.\nonumber
\end{eqnarray}
%\begin{eqnarray}\label{problem:DP}
%\max_{\mathbf{b}_{int}} && \sum_{i=1}^{NK} \left[\mathbf{b}_{int}(i) - \mathbf{b}_{rd}(i) \right] \mathbf{r}_v(i) \\
%\text{s.t.} && \sum_{i=1}^{NK} \left[\mathbf{b}_{int}(i) - \mathbf{b}_{rd}(i) \right] \leq B_{re}, \nonumber\\
%&& \left[\mathbf{b}_{int}(i) - \mathbf{b}_{rd}(i) \right] \in \{0,1 \}, \nonumber
%\end{eqnarray}
The problem in (\ref{problem:DP}) is a classical $0/1$ Knapsack Problem\cite{lew2006dynamic} and can be solved by the dynamic programming algorithm. The detailed steps of the algorithm is provided in \emph{Algorithm} \ref{table:algorithm_table}. Finally, we obtain the bits allocation result with integer constraint given by $B_{n,k}^{int \star} = \mathbf{b}_{int}(i)$ with $i = N(k-1) + n$.
\begin{remark}
The overall complexity of the dynamic programming algorithm in \emph{Algorithm} \ref{table:algorithm_table} is mainly determined by the complexity of at each iteration and the number
of iterations required. The computations within each iteration are just a few scalar additions. Moreover, the time complexity of the dynamic programming algorithm shown in \emph{Algorithm} \ref{table:algorithm_table} is $O(N K B_{re})$, and it is easy to see $B_{re} \leq \lceil NK / 2 \rceil$. Therefore, the time complexity of the algorithm is a polynomial of the scale of the system, i.e., the parameters $N, K$. Thus, this is a very low-complexity algorithm.
\end{remark}

\subsection{Power Allocation}\label{subsec:power_allocation}
%We consider optimizing the power allocation among the users in this subsection.
The power allocation in NOMA systems is more complex than that in OMA systems.
For the power allocation within one user cluster, according to the principle of NOMA, on the one hand the power allocation usually follows the criterion that $P_{n,1} \leq \cdots \leq P_{n,K}$ in order to facilitate SIC and guarantee certain level of quality of service (QoS) requirements of the users\footnote{Specifically, certain level of ESR of each user needs to be achieved by power allocation in this paper.} with poor channel conditions\cite{ZhiguoDing2016}.
On the other hand, in order to maximize the ESR, the power allocation should follow the criterion that $P_{n,1} \geq \cdots \geq P_{n,K}$.
Therefore, \cite{X.Chen_2017} proposed equal power allocation within a cluster to satisfy the above two criterions simultaneously and to reduce the complexity of the algorithm.
We will follow the same scheme, i.e., $P_{n,k} = P_n / K$, where $P_n$ is the total power allocated to the users in the $n$-th cluster.
For the power allocation among the clusters, recall that power allocation and bits allocation actually \emph{interact with each other}. Thus, in contrast to the intuitive scheme in \cite{X.Chen_2017} which was based on a given feedback bits allocation (e.g., equal feedback bit allocation among all users in the numerical result of \cite{X.Chen_2017}), we propose to \emph{jointly} optimize these two parameters by maximizing the system ESR under the premise of $P_{n,k} = P_n / K$ $\forall n, k$.
%\cite{X.Chen_2017} proposed an intuitive scheme that allocated power proportionally to the reciprocal of the interference coefficients which were obtained based on a given feedback bits allocation (e.g., equal feedback bit allocation among all users in the numerical result of \cite{X.Chen_2017}). At last, the final feedback bits allocation result follows with the obtained power allocation.
%Recall that power allocation and bits allocation actually interact with each other, thus in contrast to the intuitive scheme in \cite{X.Chen_2017}, we propose to jointly optimize these two parameters by maximizing the system ESR under the premise of $P_{n,k} = P_n / K$ $\forall n, k$.

For the tractability, we develop in Appendix \ref{appendix:LB2} another analytical lower bound of $R_{n,k}$ $\forall n, k$ as
\begin{eqnarray}\label{eq:R_nk_LB2}
R_{n,k} &\geq& R_{n,k }^{LB2} \triangleq  \Delta R_{n,k }^{LB2}   + \tilde{R}^{LB2}_{n,k}
\end{eqnarray}
with $ \Delta R_{n,k }^{LB2}  = \Bbb{E} \left[ \log_2 \left( \frac{1 + \bar{S}_{n,k}^{(1)} \left| \mathbf{h}_{n,k} \mathbf{w}_n \right|^2}{1 + \bar{S}_{n,k}^{(1)} M^{-1} ||\mathbf{h}_{n,k}||^2}  \right) \right] $, which is a negative number independent with bits allocation and no smaller than the constant $-\log_2(e) \mathbf{C} \approx -0.8327$ for any system parameters\footnote{$ \Delta R_{n,k }^{LB2}$ can be obtained exactly in closed-form. However, since it is not considered further for power allocation problem, we will not present the result.}, where $\mathbf{C}$ is the Euler’s constant \cite[8.367.2]{Gradshteyn2007}. $ \Delta R_{n,k }^{LB2} $ changes very slowly with $P_n$, thus can be identified $P_n$-independent for power allocation. In addition,
\vspace{-1mm}
\begin{eqnarray}\label{eq:tilde_R_nk_LB}
\hspace{-11mm}&&\tilde{R}^{LB2}_{n,k} = \nonumber\\
\hspace{-11mm}&& \log_2 \left( 1 + \frac{ \frac{d_{n,k}^{-\alpha} P_{n}}{\sigma_{n,k}^2 K} }{\frac{M}{M-1} + \bar{S}_{n,k}^{(2)} + \Gamma \left( \frac{2M-1}{M-1} \right) 2^{-\frac{B_{n,k}}{M-1}} \bar{S}_{n,k}^{(3)}} \right).
\end{eqnarray}
Here, $\bar{S}_{n,k}^{(1)} = \frac{d_{n,k}^{-\alpha} k P_n}{\sigma_{n,k}^2 K}$, $\bar{S}_{n,k}^{(2)} = \frac{d_{n,k}^{-\alpha} (k-1) P_n}{\sigma_{n,k}^2 K}$ and $\bar{S}_{n,k}^{(3)} = \frac{d_{n,k}^{-\alpha}}{\sigma_{n,k}^2} \sum_{i=1,i \neq n}^{N} P_{i} = \frac{d_{n,k}^{-\alpha}}{\sigma_{n,k}^2} (P - P_n)$ are obtained by substituting $P_{n,k} = P_n / K$ into $S_{n,k}^{(1)}$, $S_{n,k}^{(2)}$ and $S_{n,k}^{(3)}$ respectively. Specifically, $\bar{S}_{n,k}^{(2)} =0$ for $k=1$.

We find it still very complicated to take $\sum_{n=1}^{N}\sum_{k=1}^{K}\tilde{R}^{LB2}_{n,k} $ as the objective function to optimize power allocation. Thus, we will further simplify the problem.
We can show, with \emph{perfect} CSIT and $P_{n,k} = P_n / K$, $R_{n,k} $ for $k > 1$ can be upper bound as
\begin{eqnarray}
\hspace{-7mm}&&R_{n,k} = \log_2 \left( 1 + \frac{\frac{d_{n,k}^{-\alpha} P_{n}}{\sigma_{n,k}^2 K} |\mathbf{h}_{n,k} \mathbf{w}_n|^2}
{1 + \frac{(k-1) d_{n,k}^{-\alpha} P_{n}}{\sigma_{n,k}^2 K} |\mathbf{h}_{n,k} \mathbf{w}_n|^2} \right) \nonumber \\
\hspace{-7mm}&&< \log_2 \left( 1 + \frac{\frac{d_{n,k}^{-\alpha}}{\sigma_{n,k}^2 K} |\mathbf{h}_{n,k} \mathbf{w}_n|^2}
{\frac{(k-1) d_{n,k}^{-\alpha}}{\sigma_{n,k}^2 K} |\mathbf{h}_{n,k} \mathbf{w}_n|^2} \right) = \log_2 \left( 1 + \frac{1}{k-1} \right), \nonumber
\end{eqnarray}
which, compared to $R_{n,1}$, is a very limited amount independent with power allocation and decreasing with the user index. Moreover, the ESR of these users in each cluster is bounded as $\sum_{k=2}^K R_{n,k} < \log_2(K)$.
The corresponding rates with limited CSI feedback must be even smaller. Thus, to further consider the ergodic rate of all users within each cluster in optimization can achieve very little benefit of sum rate, but it take much more computation complexity to search for a local-optimal solution due to the non-convex property of the sum ergodic rate.
%In fact, as we can see user $(n, k)$ with $k >1$ is interference-limited. Thus, further increasing power of these users contributes very little to the sum rate of the system.
In addition, the equal power allocation within each cluster already guarantees certain transmission rate for these users.
Therefore, in the following we take the sum of the terms $\tilde{R}^{LB2}_{n,k} $ in (\ref{eq:tilde_R_nk_LB}) for $k=1$ as the objective function, i.e., $\tilde{R}_{sum}^{LB2}\; (1) \triangleq \sum_{n=1}^{N}\tilde{R}^{LB2}_{n,1} $.

Let $P_{n} = \phi_{n} P$, where the coefficients $ \phi_{n} $ ($n=1, 2, \cdots, N$) satisfy $\sum_{i=1}^N \phi_{n} = 1 $. Substituting the result of the bits allocation without integer constraint given by (\ref{eq:B_nk}) into the expression of $ \tilde{R}^{LB2}_{n,k} $, after some manipulations we have
\begin{eqnarray} \label{eq:tilde_R_nk_LB_2}
\tilde{R}^{LB2}_{n,k}  (\left\{\phi_n \right\}) = \log_2 \left( 1 +\frac{d_{n,k}^{-\alpha} \phi_{n} P}{K \sigma_{n,k}^2 D_{n,k}}
 \right),
\end{eqnarray}
where $D_{n,k} = \frac{M}{(M-1)} + \frac{d_{n,k}^{-\alpha} (k-1) \phi_n P}{\sigma_{n,k}^2 K} \\+ \Gamma \left( \frac{2M-1}{M-1} \right) 2^{-\frac{B}{NK(M-1)}} \left( \prod_{p=1}^N \prod_{q=1}^K \left( \frac{d_{p,q}^{-\alpha}}{\sigma_{p,q}^2} \right)^{\frac{1}{NK}} \right)
\\ \times \left( 1 + \frac{d_{n,k}^{-\alpha} (k-1) \phi_n P}{\sigma_{n,k}^2 K} \right)\left( \prod_{p=1}^N \left( 1 - \phi_p \right)^{\frac{1}{N}} \right)
\\ \times  \left( \prod_{p=1}^N \prod_{q=1}^K \left( \frac{1}{P} + \frac{d_{p,q}^{-\alpha} (q-1) \phi_p}{\sigma_{p,q}^2 K} \right)^{-\frac{1}{NK}}\right)$.
%\begin{eqnarray}\label{eq:D_nk}
%D_{n,k} &=& \frac{M}{(M-1)} + \frac{d_{n,k}^{-\alpha} (k-1) \phi_n P}{\sigma_{n,k}^2 K} + \Gamma \left( \frac{2M-1}{M-1} \right)2^{-\frac{B}{NK(M-1)}}\left( \prod_{p=1}^N \prod_{q=1}^K \left( \frac{d_{p,q}^{-\alpha}}{\sigma_{p,q}^2} \right)^{\frac{1}{NK}} \right) \nonumber\\
%&& \hspace{-1.5cm } \times \left( 1 + \frac{d_{n,k}^{-\alpha} (k-1) \phi_n P}{\sigma_{n,k}^2 K} \right) \left( \prod_{p=1}^N \left( 1 - \phi_p \right)^{\frac{1}{N}} \right)\left( \prod_{p=1}^N \prod_{q=1}^K \left( \frac{1}{P} + \frac{d_{p,q}^{-\alpha} (q-1) \phi_p}{\sigma_{p,q}^2 K} \right)^{-\frac{1}{NK}}\right). ~ ~
%\end{eqnarray}
Then, our power allocation problem becomes
\begin{eqnarray}\label{eq:power_optimization}
\min_{\left\{\phi_n \right\}} && - \tilde{R}_{sum}^{LB2}\; (1)  =  - \sum_{n=1}^{N}\tilde{R}^{LB2}_{n,1} (\left\{\phi_n \right\}) \\
\text{s.t.} && \sum_{n=1}^{N} \phi_n \leq 1,  ~ \phi_{n} \geq 0. \nonumber
\end{eqnarray}
It is not difficult to check that the problem in (\ref{eq:power_optimization}) is still a non-convex optimization problem. Thus, we will seek for a sub-optimal (local-optimal) solution, which is given by the following theorem.
\begin{theorem}\label{theorem:power_alloc}
Given $N$ and when transmit power $P$ is large enough to support the $N$ clusters of NOMA users, a sub-optimal solution to the problem (\ref{eq:power_optimization}) is given by
\begin{eqnarray}\label{eq:phi_n_res_final}
\hspace{-10mm}&&\phi_n^{\star} = \frac{1}{N} - \frac{K}{N} \Bigg[ \frac{M}{(M-1)P} + \Gamma \left( \frac{2M-1}{M-1} \right) \frac{2^{-\frac{B}{NK(M-1)}}}{P}\nonumber\\
\hspace{-10mm}&&\times  \bigg( \prod_{p=1}^N \prod_{q=1}^K \left( \frac{d_{p,q}^{-\alpha}}{\sigma_{p,q}^2} \right)^{\frac{1}{NK}} \bigg) C^{\star} \Bigg] \sum_{i=1,i \neq n}^N \bigg( \frac{\sigma_{n,1}^2}{d_{n,1}^{-\alpha}} - \frac{\sigma_{i,1}^2}{d_{i,1}^{-\alpha}} \bigg) \hspace{-0.5mm},
\end{eqnarray}
where $C^{\star}$ is the solution of the following equation of variable $C$ given as
%\begin{eqnarray}\label{eq:equation_C}
%C = \left( \prod_{p=1}^N \left( 1 - \phi_p^{\star} \right)^{\frac{1}{N}} \right)  \prod_{p=1}^N \prod_{q=1}^K \Big( \frac{1}{P} + \frac{d_{p,q}^{-\alpha} (q-1) \phi_p^{\star}}{\sigma_{p,q}^2 K} \Big)^{-\frac{1}{NK}},
%\end{eqnarray}
\begin{eqnarray}\label{eq:equation_C}
\hspace{-7mm}&&C =\\
\hspace{-7mm}&&P^{\frac{1}{K}} \left( \prod_{p=1}^N \left( 1 - \phi_p^{\star} \right)^{\frac{1}{N}} \right)  \prod_{p=1}^N \prod_{q=2}^K \Big( \frac{1}{P} + \frac{d_{p,q}^{-\alpha} (q-1) \phi_p^{\star}}{\sigma_{p,q}^2 K} \Big)^{-\frac{1}{NK}} \nonumber
\end{eqnarray}
with $\phi_p^{\star}$ $ (p=1,2, \cdots, N)$ being a function of $C$ given by (\ref{eq:phi_n_res_final}). The desired solution of $C^{\star} $ can be obtained by solving the equation of $C$ in (\ref{eq:equation_C}) with any numerical method. Moreover, the power allocation satisfies $\phi_p^{\star} \geq \phi_2^{\star} \geq \cdots \geq \phi_N^{\star}  $.
\end{theorem}
\begin{proof}
See Appendix \ref{proof:theorem_power_alloc}.
\end{proof}
\begin{remark}\label{remark:power_alloc}
Generally, there are multiple roots with the equation (\ref{eq:equation_C}). The desired solution of $C^{\star} $ is chosen as the one makes the resulting $\phi_n^{\star} \in [0,1]$ practical. If there are more than one roots of $C^{\star} $, the desired one can be determined by checking the obtained $\tilde{R}_{sum}^{LB2} \; (1)$.
\end{remark}

According to the power allocation in Theorem \ref{theorem:power_alloc}, it is easy to see, given $N$ the total transmit power $P$ needs to be large enough to support $N$ clusters of NOMA users based on our framework described above. Otherwise, $\phi_N^{\star} $ will become a negative number. When $P$ is not large enough to support $N$ clusters of users, there are two ways to circumvent this problem: (a) One can reduce the number of clusters by one and redo the user clustering of the $NK$ users according to (\ref{eq:cluster_criteria1}) and (\ref{eq:cluster_criteria2}) until $\phi_N^{\star} >0 $ holds for the number of clusters $N$; or (b) remove the users in cluster $N$ from the time-frequency RB considered and multiplex these users on some other time-frequency RB.

\subsection{Asymptotic Results of the Proposed Power and Feedback Bits Allocation}\label{subsec:asymptotic results}
In order to provide insights for the system design, we now pursue an asymptotic analysis on our joint optimization of power and feedback bits allocation that are obtained above. First, the power allocation in the high power region and in the high CSI quantization accuracy scenario are given respectively as follows.
\begin{corollary}\label{corollary:asy_power}
When $B \rightarrow +\infty$ with fixed and finite $P $, our solution in \emph{Theorem} \ref{theorem:power_alloc} tends to be water-filling type solution for the system with perfect CSIT. Moreover, when $P \rightarrow + \infty$ with fixed and finite $B$, our power allocation solution tends to be equal power allocation among all users.
\end{corollary}
\begin{proof}
First, it can be observed from (\ref{eq:phi_n_res_final}) that, when $B \rightarrow +\infty$ with fixed and finite $P $, our solution turns to be $\phi_n^{\star} \rightarrow \frac{1}{N} - \frac{K}{NP}  \frac{M}{(M-1)} \sum_{i=1,i \neq n}^N \left( \frac{\sigma_{n,1}^2}{d_{n,1}^{-\alpha}} - \frac{\sigma_{i,1}^2}{d_{i,1}^{-\alpha}} \right)$, which is exactly the water-filling type solution for the system with perfect CSIT. Moreover,
%it is easy to be observed from (\ref{eq:equation_C}) that,
when $P \rightarrow +\infty$ with finite $B$,
\begin{eqnarray}\label{eq:C_limit_P}
C^{\star} \rightarrow
\frac
{P^{\frac{1}{K}} \left( \prod_{p=1}^N \left( 1 - \phi_p^{\star} \right)^{\frac{1}{N}} \right)}
{\left[ \left(\prod_{p=1}^N {\phi_p^{\star}}\right)^{1- \frac{1}{K}}  \prod_{p=1}^N  \prod_{q=1}^K  \left( \frac{ (q-1) d_{p,q}^{-\alpha}}{\sigma_{p,q}^2 K} \right)^{\frac{1}{NK}}\right]} \hspace{-0.5mm}. \hspace{-1mm}
\end{eqnarray}
Then, the equal power allocation follows by substituting (\ref{eq:C_limit_P}) into (\ref{eq:phi_n_res_final}).
\end{proof}

\begin{theorem}\label{theorem:Bn_k_asy}
As the total transmit power $P$ increases from the medium power region to high power region, with finite total bandwidth of feedback channels $B$ and without integer constraint, the optimal number of feedback bits allocated to the nearest user $(n, 1)$ in each cluster is scaled approximately \emph{increasingly} at rate $\left(M -1 \right) \left(1- \frac{1}{K} \right) \log_2\left( P\right)$, whilst the optimal number of feedback bits allocated to the other users is scaled approximately \emph{decreasingly} at rate $\left(M -1 \right) \left(1- \frac{1}{K} \right) \log_2\left( P\right)$.
Specifically, as $P$ increases we have
%we have the optimal number of feedback bits as $B_{n, k}^{\star} \approx  \max \{ \tilde{B}_{n, k}^{\star} , 0\}$ with
\begin{eqnarray}\label{eq:Bn_1_asy}
\hspace{-0.8cm}&&B_{n, 1}^{\star} \approx \tilde{B}_{n, 1}^{\star} = \frac{B}{NK} + \left(M -1 \right) \left(1- \frac{1}{K} \right) \log_2\left( \frac{P}{NK} \right) \nonumber\\
\hspace{-0.8cm}&&\hspace{0.8cm} + (M-1) \log_2\left( \frac{d_{n,1}^{ - \alpha}}{\sigma_{n,1}^2} \right)
- \frac{M-1}{NK} \sum_{i =1}^{N } \log_2\left( \frac{d_{i,1}^{ - \alpha}}{\sigma_{i,1}^2} \right) \nonumber\\
\hspace{-0.8cm}&&\hspace{0.8cm} + \frac{M-1}{K}\sum_{l =1}^{K -1} \log_2(l),~ \forall n;\\
\label{eq:Bn_k_asy}
\hspace{-0.8cm}&&B_{n, k}^{\star} \approx \tilde{B}_{n, k}^{\star} = \frac{B}{NK} - \frac{M-1}{K} \log_2\left( \frac{P}{NK} \right)\nonumber\\
\hspace{-0.8cm}&&\hspace{0.8cm}- \frac{M-1}{NK} \sum_{i =1}^{N } \log_2\left( \frac{d_{i,1}^{ - \alpha}}{\sigma_{i,1}^2} \right) + \frac{M-1}{K}\sum_{l =1}^{K -1} \log_2(l) \nonumber\\
\hspace{-0.8cm}&&\hspace{0.8cm}-(M-1) \log_2(k-1), ~\text{for}~ k \geq 2~ \text{and} ~ \forall n.
\end{eqnarray}
\end{theorem}
\begin{proof}
According to \emph{Corollary} \ref{corollary:asy_power}, the power allocation of our joint optimization method tends to be equal power allocation in the high power region. Then, the theorem can be proved easily by substituting the result of $P_n= P/N$ into $\bar{S}_{n,k}^{(2)} $, $\bar{S}_{n,k}^{(3)} $ and (\ref{eq:B_nk}) with some manipulations.
\end{proof}
\begin{remark}
We see from this theorem that the feedback bits allocated to each user by our method varies differently as $P$ increases, whereas the feedback bits allocation by the method in \cite{X.Chen_2017} remains invariant as $P$ varies.
Moreover, this theorem shows, as $P$ increases from the medium power region to high power region, the nearest (strongest) user in each cluster is allocated more feedback bits than the other users in the cluster, and among these nearest users of all clusters, the stronger the user is the more feedback bits is allocated. In contrast, for the users other than the strongest user in each cluster (i.e., user $(n,k)$ for $k \geq 2$), the weaker the user is the more feedback bits is allocated. These insights can also be verified by the numerical results in Section \ref{sec:numerical_results}.
\end{remark}
Moreover, we have the following result for finite $B$.
\begin{corollary}
With finite $B$ and without integer constraint, as $P $ goes to $+\infty$, the optimal number of feedback bits allocated to the users other than the nearest user in each cluster reduces to zero, whilst the optimal number of feedback bits allocated to user $(n, 1) $ of each cluster converges to a finite number. Specifically, we have $B_{n,1}^{\star} \approx \min\{ \hat{B}_{n,1}^{\star}, \tilde{B}_{n,1}^{\star}  \}$ and $B_{n,k}^{\star} \approx \min\{ 0, \tilde{B}_{n,k}^{\star} \}$ for $k \geq 2$, where each $\tilde{B}_{n, k}^{\star} $ is defined by (\ref{eq:Bn_1_asy}) or (\ref{eq:Bn_k_asy}) in \emph{Theorem} \ref{theorem:Bn_k_asy} and
\begin{eqnarray}\label{eq:hat_Bn1}
\hspace{-7mm}&&\hat{B}_{n,1}^{\star}\\
\hspace{-7mm}&&= \frac{B}{N} + (M -1) \log_2 \left( \frac{d_{n,1}^{ - \alpha}}{\sigma_{n,1}^2} \right)
-  \frac{M-1}{N} \sum_{i =1}^{N } \log_2 \left( \frac{d_{i,1}^{ - \alpha}}{\sigma_{i,1}^2} \right) \hspace{-1mm}.\nonumber
\end{eqnarray}
\end{corollary}
\begin{proof}
As $P \rightarrow +\infty$ with finite $B$, it is obvious that $B_{n,k}^{\star}\rightarrow 0 $ and $\tilde{B}_{n,k}^{\star}  \rightarrow 0$ for $k\geq 2$ and $\forall n$.
The result of $\hat{B}_{n,1}^{\star} $ in (\ref{eq:hat_Bn1}) is obtained by substituting the expression of $\tilde{B}_{i, 1}^{\star} $ in (\ref{eq:Bn_1_asy}) into the constraint $\sum_{ i=1}^{N} \tilde{B}_{i, 1}^{\star} = B$, which completes the proof of the corollary.
\end{proof}
\begin{remark}\label{remark:last}
It is easy to check that $\{\tilde{B}_{n, k}^{\star}\}$ satisfies the constraint on the total bandwidth of feedback channels. According to (\ref{eq:Bn_1_asy}) and the user clustering criteria in (\ref{eq:cluster_criteria1}) and (\ref{eq:cluster_criteria2}), it is easy to see when $P$ is large enough, our bits allocation results in $B_{1, 1}^{\star} \geq B_{2, 1}^{\star} \geq \cdots, \geq B_{N, 1}^{\star} $.
\end{remark}

In the following, we will show for our joint power and feedback bits allocation method, the proposed user clustering (denoted as $\mathfrak{C}$) in general is the optimal in terms of ESR among all possible user clustering methods. First note that, it is obvious that, if a different user clustering (denoted as $\tilde{\mathfrak{C}}$) is obtained from $\mathfrak{C}$ by exchanging users $(i, p )$ and $(j, q)$ for $i \neq j$ and $p \neq q$ where $(i, p )$ and $(j, q)$ follow (\ref{eq:cluster_criteria1}) and (\ref{eq:cluster_criteria2}), the ESR of $\tilde{\mathfrak{C}}$ will be reduced compared to that of $\mathfrak{C}$. Therefore, we can only focus on the user clustering obtained from $\mathfrak{C}$ by fixing the strongest user in each cluster and exchanging a few pairs of users with indices $(i, k)$ and $(j, k)$ for $i \neq j$ and $k \geq 2$.
%
%we can compare $\mathfrak{C}$ with the other user clusterings with the same strongest user in each cluster only, since exchanging the strongest users of any two clusters of $\mathfrak{C}$ is equivalent to exchanging all the other users within these two clusters.
%First note we can fix the strongest user within each cluster and only compare ours with the others which exchange the users other than the strongest users in each cluster (of our user clustering), since exchanging the strongest users of any two clusters is equivalent to exchanging all the other users of this two clusters.
%Moreover, it is easy to see we can safely consider the comparisons between ours with the clustering methods which is obtained from ours by exchanging the users with the same user index $k$ in each cluster.
Then, we first consider $\tilde{\mathfrak{C}}$  obtained from $\mathfrak{C}$ by exchanging any two users $(n,k)$ and $(m, k)$ with $m \leq n$ (i.e., user $(n,k)$ and user $(m,k)$ in $\mathfrak{C}$ become user $(m,k)$ and user $(n,k)$ respectively in $\tilde{\mathfrak{C}}$).
According to Theorem \ref{theorem:B_without_integer} and \ref{theorem:power_alloc}, it is easy to see the power allocated to user $(n,k)$ and $(m, k)$ originally in $\mathfrak{C}$ are exchanged and the feedback bits allocated to these two users are varied. However, the power and bits allocated to the other users do not change. Thus, the ergodic rate of any user other than the above two in $\tilde{\mathfrak{C}}$ is the same as that of the corresponding user in $\mathfrak{C}$.

We denote the ergodic rates of users $(m,k)$ and $(n,k)$ in $\tilde{\mathfrak{C}}$  as $\tilde{R}_{m,k}$ and $\tilde{R}_{n,k}$ respectively, which are obtained by substituting the power and feedback bits allocation results (without integer constraint) given by Theorem \ref{theorem:B_without_integer} and \ref{theorem:power_alloc} into the expression of SINR in (\ref{eq:gamma_nk_gt1}).
Compared with $\mathfrak{C}$, since in $\tilde{\mathfrak{C}}$ the stronger user is allocated less power and the weaker user is allocated with more power (i.e., $P_n < P_m$), it is reasonable for one to anticipate the following result $R_{m, k} + R_{n, k} > \tilde{R}_{m, k} + \tilde{R}_{n, k} $ holds in general. In addition, the larger the gap of power allocation $P_m - P_n \geq  0$ is or equivalently the larger the distance between the two indices $n -m $ is, the larger the gap of ESR $(R_{m, k} + R_{n, k})  - (\tilde{R}_{m, k} + \tilde{R}_{n, k} )$ becomes. Moreover, if $\tilde{\mathfrak{C}}$ is obtained from $\mathfrak{C}$ by exchanging more such user pairs, the rate gap between two user clusterings $\mathfrak{C}$ and $\tilde{\mathfrak{C}}$ will be increased.
We find after some tedious manipulations that, the \emph{exact} closed-form results of $R_{i, k}$ and $\tilde{R}_{i, k}$ for $i =m ,n$ are necessary for the comparison between $\mathfrak{C}$ and $\tilde{\mathfrak{C}}$ \emph{in theory}.
But unfortunately, as we have noted in Section \ref{sec:actual_performance} that it is very difficult if not impossible to obtain the {exact} closed-form results. Thus, the ESR comparison between two user clusterings in theory is intractable. We have to turn to numerical experiments.
%%%%%%%%%%%%%%%%%%%%%%%%%%%%%%%%%%%%%%%%%%%%%%%%%%%%%%%%%%%%%%%%%%%%%%%%%%%%%%%%%%%%%%%%%%%%%%%%%%%%%%%%%%%%%%%%%

For verification, we compare the ESRs of different user clustering methods with the system parameters given by $M = 6$, $B = 72$ bits, $N = 3$ and $K = 2$.
%All fast-fading channels follow the same distribution as that presented in Section \ref{sec:actual_performance} and \ref{sec:Bit_Power_Allocations}.
The AWGN at all users are with the equal power, i.e., $\sigma^2_{n,k} = -50$ dBm $\forall n, k$.
The distance parameters are given by $\mathbf{D}_1 = \Big[ \begin{smallmatrix}  25  & 35 \\ 27 & 37 \\  29 & 39 \\ \end{smallmatrix} \Big]$ with the $(n,k)$-th element denoting $d_{n,k}$ in meters. Fig. \ref{fig:plot_exchange} shows the ESR gaps between our user clustering the other methods, where $R_{\text{sum}} \left[\left( p, q, r\right) \right]$ with $p, q, r \in \{1,2,3\} $ denotes the ESR of $\tilde{\mathfrak{C}}$, and the users $(1,2), (2, 2), (3, 2)$ in $\tilde{\mathfrak{C}}$ corresponds to the users $(p, 2), (q, 2), (r, 2)$ originally in $\mathfrak{C}$ respectively.
Specifically, $R_{\text{sum}} \left[\left( 1, 2, 3 \right) \right]$ denotes ESR of our method. As anticipated, it is easy to see from Fig. \ref{fig:plot_exchange} that the performance of our method is the optimal. Moreover, we can observe the following order of the ESR that $\{ R_{\text{sum}} \left[\left( 1, 2,3 \right) \right] \} > \{ R_{\text{sum}} \left[\left( 2, 1, 3 \right) \right], R_{\text{sum}} \left[\left( 1, 3, 2 \right) \right] \}  > \{R_{\text{sum}} \left[\left( 3, 1, 2 \right) \right], R_{\text{sum}} \left[\left( 2, 3, 1\right) \right]  \}> \{R_{\text{sum}} \left[\left( 3, 2, 1 \right) \right] \}$, which complies with the intuitive analysis above. However, the relative amounts of ESRs within a ESR set depend on the actual system parameters in general, and thus are \emph{not comparable}.
In the following section, we will validate the obtained analytical results and our proposed scheme by more numerical results.

%%%%%%%%%%%%%%%%%%%%%%%%%%%%%%%%%%%%%%%%%%%%%%%%%%%%%%%%%%%%%%%%%%%%%%%%%%%%%%%%
%%%%%%%%%%%%%%%%%%%%%%%%%%%%%%%%%%%%%%%%%%%%%%%%%%%%%%%%%%%%%%%%%%%%%%%%%%%%%
%\begin{figure}[t]
%\centering
%\includegraphics[width= 0.75\columnwidth]{plot_compare_joint_optimization.eps}
%\caption{Comparisons of the simulation and analytical results of the ESR v.s. $P$ for our proposed scheme and those in \cite{X.Chen_2017}.}
%\label{fig:plot_compare_joint_optimization}
%\end{figure}
%
%\begin{figure}[t]
%\centering
%\includegraphics[width= 0.75\columnwidth]{plot_compare_bit_allocate_my_with_cxm.eps}
%\caption{Comparisons of the ESR the different feedback bits allocation methods.}
%\label{fig:plot_compare_bit_allocate_my_with_cxm}
%\end{figure}

%\begin{figure}
%\centering{}%
%\begin{minipage}[t]{0.48\textwidth}%
%\begin{center}
%\includegraphics[clip,width=82mm]{exchange_user_123.eps}
%\par\end{center}
%\caption{\label{fig:plot_exchange}{\small{}The ESR losses of some of the other user clustering methods compared with our user clustering method.}}
%\end{minipage}
%\begin{minipage}[t]{0.48\textwidth}%
%\begin{center}
%\includegraphics[clip,width=82mm]{plot_compare_joint_optimization.eps}
%\par\end{center}
%\caption{\label{fig:plot_compare_joint_optimization}{\small{}Comparisons of the simulation and analytical results of the ESR v.s. $P$ for our proposed scheme and those in \cite{X.Chen_2017}.}}
%\end{minipage}\hfill{}%
%\vspace{-0.9cm}
%\end{figure}

\begin{figure}[t]
\centering
\includegraphics[width=0.95\columnwidth]{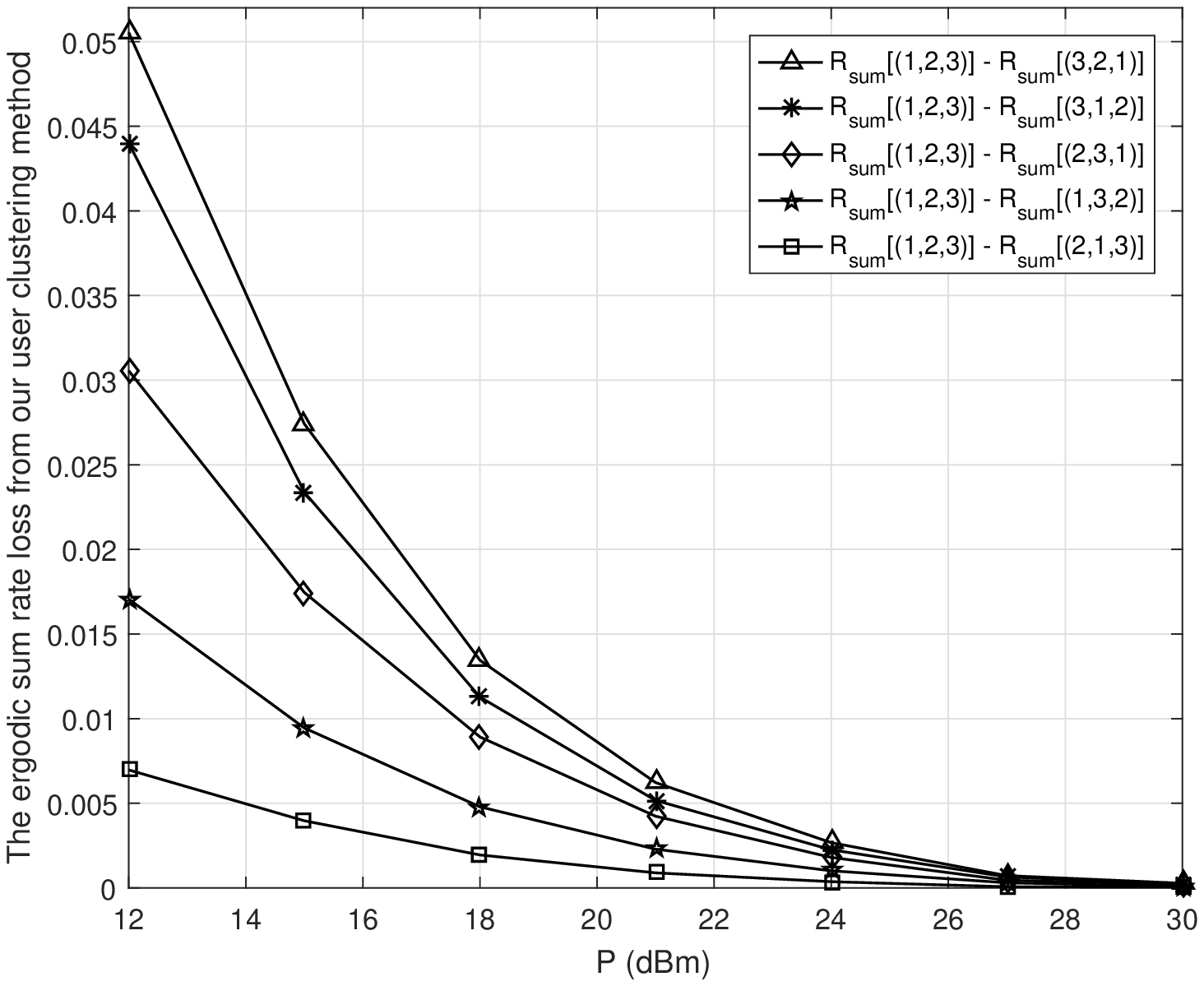}
\caption{\label{fig:plot_exchange}{\small{}The ESR losses of some of the other user clustering methods compared with our user clustering method.}}
\vspace{-9mm}
\end{figure}

\begin{figure}[t]
\centering
\includegraphics[width=0.95\columnwidth]{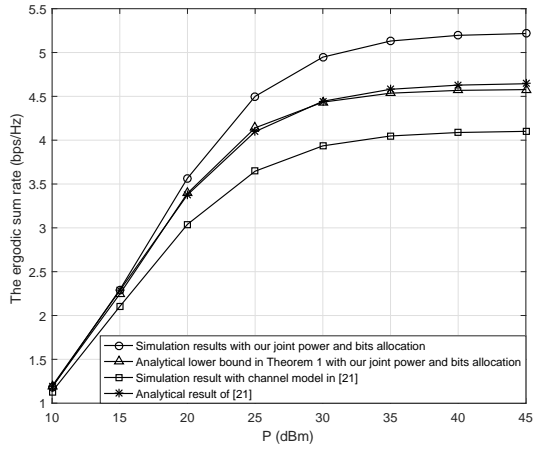}
\caption{\label{fig:plot_compare_joint_optimization}{\small{}Comparisons of the simulation and analytical results of the ESR v.s. $P$ for our proposed scheme and those in \cite{X.Chen_2017}.}}
\vspace{-6mm}
\end{figure}

%%%%%%%%%%%%%%%%%%%%%%%%%%%%%%%%%%%%%%%%%%%%%%%%%%%%%%%%%%%%%%%%%%%%%%%%%%

\section{Numerical Results}\label{sec:numerical_results}

%For the numerical results, we consider the systems with $M = 6$ and $6$ scheduled users with $N = 3$ and $K = 2$. All fast-fading channels follow the same distribution as that presented in Section \ref{sec:actual_performance} and \ref{sec:Bit_Power_Allocations}. The AWGN at all users are with the equal power, i.e., $\sigma^2_{n,k} = -50$ dBm $\forall n, k$.
%We consider two different sets of distance parameters given as $\mathbf{D}_1 = \Big[ \begin{smallmatrix}  25  & 35 \\ 27 & 37 \\  29 & 39 \\ \end{smallmatrix} \Big]$ and $\mathbf{D}_2 = \Big[ \begin{smallmatrix}  10  & 35 \\ 12& 37 \\  14 & 39 \\ \end{smallmatrix} \Big]$, where $\mathbf{D}_i $ is a $N \times K$ matrix with the $(n,k)$-th element denoting $d_{n,k}$ in meters.
For the numerical results, we consider the systems with two different sets of distance parameters given as $\mathbf{D}_1 $ of Fig. \ref{fig:plot_exchange} and $\mathbf{D}_2 = \Big[ \begin{smallmatrix}  10  & 35 \\ 12& 37 \\  14 & 39 \\ \end{smallmatrix} \Big]$ and the different number of $B$s. The other major system parameters are the same as those of Fig. \ref{fig:plot_exchange}.
Specifically, the systems associated with Fig. \ref{fig:plot_compare_joint_optimization} - Fig. \ref{fig:plot_compare_bit_allocate_my_with_cxm_bit_num} and the systems associated with Fig. \ref{fig:plot_sum_rate_vs_P} - Fig. \ref{fig:plot_sum_rate_vs_B} are with $\mathbf{D}_1$ and $\mathbf{D}_2 = \Big[ \begin{smallmatrix}  10  & 35 \\ 12& 37 \\  14 & 39 \\ \end{smallmatrix} \Big]$ respectively. The path-loss exponent $\alpha = 4$.
Moreover, if the total transmit power $P$ can not support a given number of clusters (i.e., $N=3$ at the beginning of system design), we will follow the way of (a) in the design as we have described after Remark \ref{remark:power_alloc}. The ESR curves corresponding to the obtained analytical lower bound in Theorem \ref{theorem:R_nk_LB} will be notified specifically. Otherwise, the ESR curves are obtained from simulations.

Fig. \ref{fig:plot_compare_joint_optimization} shows the simulation results of the ESR v.s $P$ for our proposed joint bits and power allocation method and the corresponding analytical lower bound given by \emph{Theorem} \ref{theorem:R_nk_LB} with $B = 42$ bits. For comparisons, we also plot the simulation and analytical results with the quantized CSI model in \cite{X.Chen_2017}.
We can see from the figure, the simulation results illustrate the fact we have explained in \emph{Remark} \ref{remark:2}, i.e., the channel model for limited feedback used in \cite{X.Chen_2017} actually underestimates the performance of the system with the well recognized practical channel model used in many previous works \cite{Jindal06,Liu2016,Q.Yang_2017}. However, we can see in the figure that the analytical ergodic rate result derived in \cite{X.Chen_2017} is right over the corresponding results of the simulations. This illustrates the analysis method in \cite{X.Chen_2017} to ignore the statistical dependence between the useful signal and MU interference signals received at each user results in \emph{overestimation} of the actual performance.
The big gap between analytical and simulated results demonstrates the impact of the aforementioned dependent relationships can not be ignored in the analysis and design of the practical systems.
Moreover, as transmit power varies, our obtained analytical lower bound can track the real values well.

Fig. \ref{fig:plot_compare_bit_allocate_my_with_cxm} compares the ESR as a function of $P$ achieved by our proposed bits allocation method with those of the equal bits allocation and the bits allocation method proposed in \cite{X.Chen_2017} for the systems with different $B$s. All schemes employ the quantized channel model in this paper and the equal power allocation among all users. We can see our bits allocation method performs the best of all. And the performance gaps between our method and the other two both increase when $B$ becomes large, which illustrates the advantage of our method.

%%%%%%%%%%%%%%%%%%%%%%%%%%%%%%%%%%%%%%%%%%%%%%%%%%%%%%%%%%%%%%%%%%%%%%%%%%%%%%
%\begin{figure}[t]
%\centering
%\includegraphics[width= 0.75\columnwidth]{plot_compare_bit_allocate_my_with_cxm_bit_num.eps}
%\caption{Comparisons of the number of feedback bits allocated to each user for different bits allocation methods.}
%\label{fig:plot_compare_bit_allocate_my_with_cxm_bit_num}
%\end{figure}
%
%\begin{figure}[t]
%\centering
%\includegraphics[width= 0.75\columnwidth]{plot_sum_rate_vs_P.eps}
%\caption{Comparison of the system ESR as a function of $P$ with different power allocation methods.}
%\label{fig:plot_sum_rate_vs_P}
%\end{figure}

%\begin{figure}
%\centering{}%
%\begin{minipage}[t]{0.48\textwidth}%
%\begin{center}
%\includegraphics[clip,width=85mm]{plot_compare_bit_allocate_my_with_cxm.eps}
%\par\end{center}
%\caption{\label{fig:plot_compare_bit_allocate_my_with_cxm}{\small{}Comparisons of the ESR v.s $P$ for different feedback bits allocation methods.}}
%\end{minipage}
%\begin{minipage}[t]{0.48\textwidth}%
%\begin{center}
%\includegraphics[clip,width=85mm]{plot_compare_bit_allocate_my_with_cxm_bit_num.eps}
%\par\end{center}
%\caption{\label{fig:plot_compare_bit_allocate_my_with_cxm_bit_num}{\small{}Comparisons of the number of feedback bits allocated to each user for different bits allocation methods.}}
%\end{minipage}\hfill{}%
%\vspace{-0.9cm}
%\end{figure}

\begin{figure}[t]
\centering
\includegraphics[width=0.95\columnwidth]{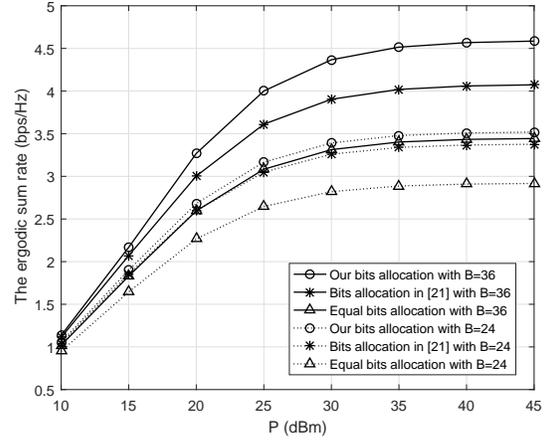}
\caption{\label{fig:plot_compare_bit_allocate_my_with_cxm}{\small{}Comparisons of the ESR v.s $P$ for different feedback bits allocation methods.}}
\vspace{-9mm}
\end{figure}

\begin{figure}[t]
\centering
\includegraphics[width=0.95\columnwidth]{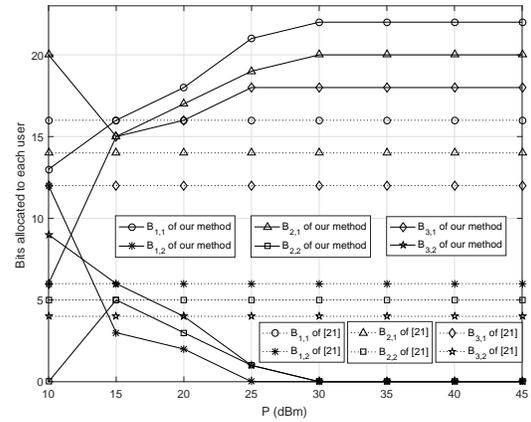}
\caption{\label{fig:plot_compare_bit_allocate_my_with_cxm_bit_num}{\small{}Comparisons of the number of feedback bits allocated to each user for different bits allocation methods.}}
\vspace{-6mm}
\end{figure}

%%%%%%%%%%%%%%%%%%%%%%%%%%%%%%%%%%%%%%%%%%%%%%%%%%%%%%%%%%%%%%%%%%%%%%%%%%%%%%%

Fig. \ref{fig:plot_compare_bit_allocate_my_with_cxm_bit_num} shows the number of feedback bits allocated to each user obtained from our method and that obtained from the method in \cite{X.Chen_2017} with the total feedback channel bandwidth of $B = 60$ bits.
We can see that, as $P$ increases our method allocates more bits to the nearest users (to the BS) in each cluster, and allocates less bits to the farther users. In very high power region, our method does not allocate feedback bits for the farther user in each cluster. It is not difficult to check the results show in this figure are consistent with our analysis in Subsection \ref{subsec:asymptotic results}.
The results also can be explained by the qualitative analysis as follows. On the one hand, increasing the quantization accuracy can enhance the strength of each user's transmit signal, and the increasing of the near users's signal strength is more effective in enhancing the transmission rate than increasing the farther users' signal strength. On the other hand, the system is interference-limited in high power region. Increasing the quantization accuracy of user $(n,k)$'s CSI can reduce the inter-cluster interference to all users $(m, k)$ $\forall m\neq n$ and $\forall k$, but only reduces the intra-cluster interference to users $(n , j) $ for $ j > k $ due to SIC. Thus, with equal power allocation within each cluster, increasing the quantization accuracy of user $(n , k) $'s CSI is more effective in reducing the MU interference than increasing the quantization accuracy of user $(n , j) $'s CSI for $ j > k $.
In contrast, the bits allocation by the method of \cite{X.Chen_2017} does not change with the increase of $P$.

Fig. \ref{fig:plot_sum_rate_vs_P} shows the ESR as a function of $P$ for different total bandwidths of feedback channels achieved respectively by our proposed joint optimization method, our bits allocation method with equal power allocation among all users, and the equal power and bits allocation method. As can be seen from the figure, our proposed power allocation converges to the equal power allocation as $P$ goes large, which is consistent with our analysis in \emph{Corollary} \ref{corollary:asy_power}. However, there is always a gap between the performance of equal power and bits allocation and those of the other two.

Fig. \ref{fig:plot_sum_rate_vs_B} compares the ESR as a function of $B$ achieved by our joint optimization method, our bits allocation method with equal power allocation and equal power and bits allocation method. As the baseline, we also plot the performance of the OMA transmission, where, for fair comparisons, the groups of the nearer users and the farther users share a fraction of $1/2$ time resource using linear ZFBF with equal power allocation for transmission. As can be seen from the figure, the performance of our joint bits and power allocation method is the best of all. The performances of all NOMA transmission schemes are better than that of OMA scheme for the system setting considered.
Moreover, we can observe that, the performance gap between our joint optimization scheme and the equal bits allocation scheme and also the performance gap between our joint optimization scheme and the OMA scheme both increase as the accuracy of channel quantization ($B$) increases.
However, the advantage of our power allocation method over the equal power allocation method decreases as the accuracy of channel quantization increases.

%%%%%%%%%%%%%%%%%%%%%%%%%%%%%%%%%%%%%%%%%%%%%%%%%%%%%%%%%%%%%%%%%%%%%%%%%%%%%%%%%%%%%%%%%
%\begin{figure}
%\centering{}%
%\begin{minipage}[t]{0.48\textwidth}%
%\begin{center}
%\includegraphics[clip,width=85mm]{plot_sum_rate_vs_P.eps}
%\par\end{center}
%\caption{\label{fig:plot_sum_rate_vs_P}{\small{}Comparison of the system ESRs as a function of $P$ with different power allocation methods.}}
%\end{minipage}
%\begin{minipage}[t]{0.48\textwidth}%
%\begin{center}
%\includegraphics[clip,width=85mm]{plot_sum_rate_vs_B.eps}
%\par\end{center}
%\caption{\label{fig:plot_sum_rate_vs_B}{\small{}Comparisons of the ESRs as a function of $B$ achieved by different power and feedback bits allocation methods.}}
%%\caption{\label{fig:plot_sum_rate_vs_B}{\small{}Comparisons of the ESRs as a function of $B$ achieved by our joint optimization method, our bits allocation method with equal power allocation, the equal power and bits allocation method and also OMA transmission.}}
%\end{minipage}\hfill{}%
%\end{figure}

\begin{figure}[t]
\centering
\includegraphics[width=0.95\columnwidth]{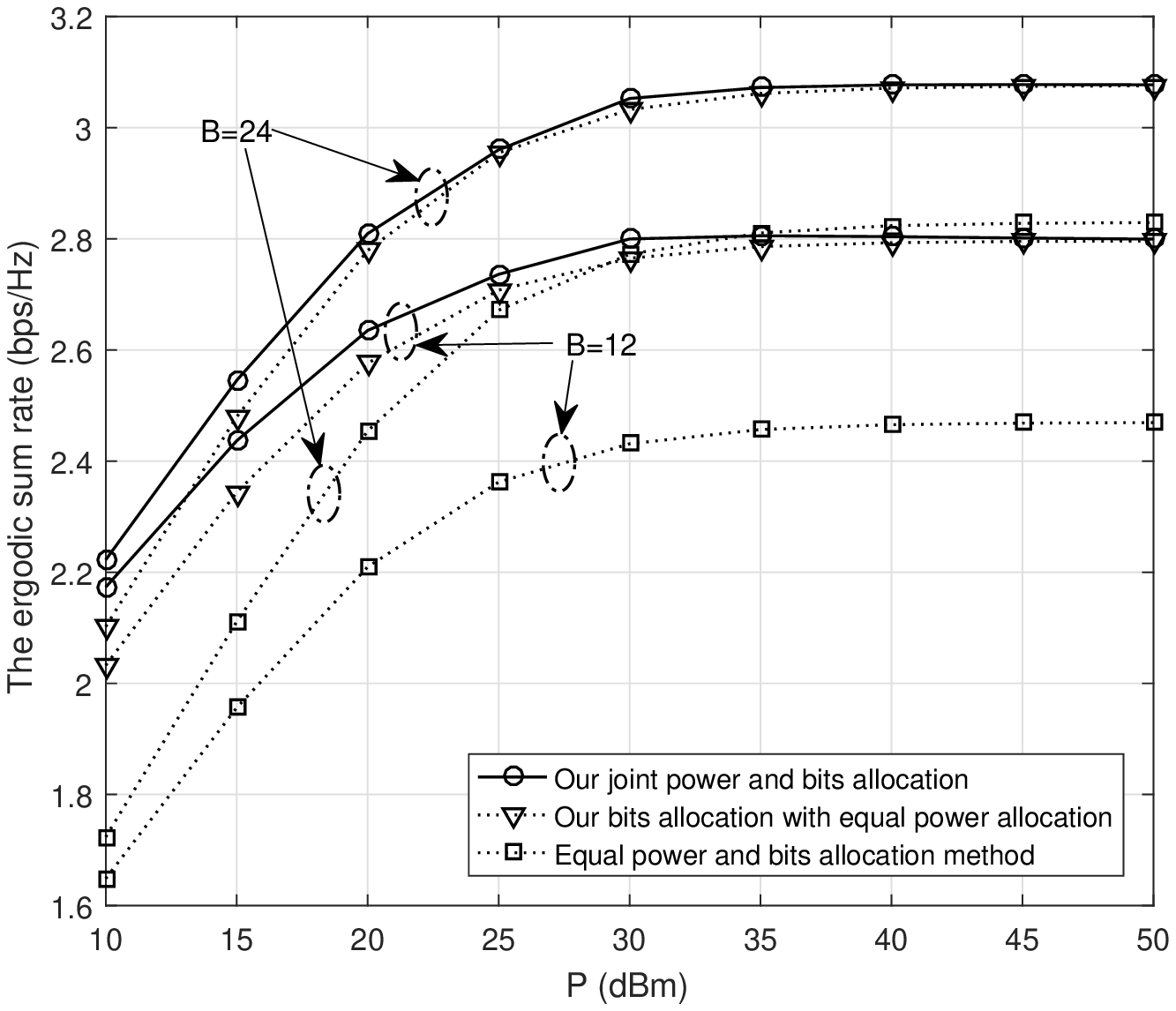}
\caption{\label{fig:plot_sum_rate_vs_P}{\small{}Comparison of the system ESRs as a function of $P$ with different power allocation methods.}}
\vspace{-9mm}
\end{figure}

\begin{figure}[t]
\centering
\includegraphics[width=0.95\columnwidth]{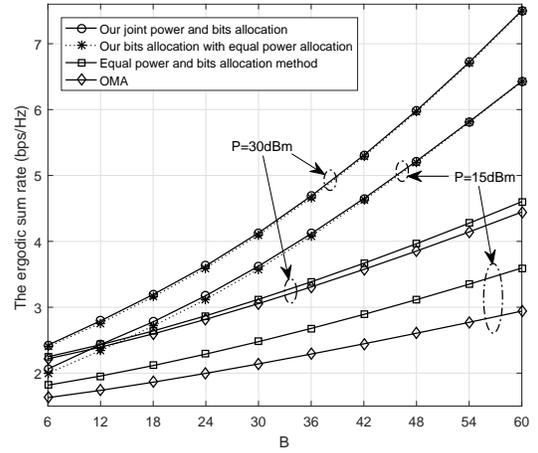}
\caption{\label{fig:plot_sum_rate_vs_B}{\small{}Comparisons of the ESRs as a function of $B$ achieved by different power and feedback bits allocation methods.}}
\vspace{-6mm}
\end{figure}

%\begin{figure}[t]
%\centering
%\includegraphics[width= 0.7\columnwidth]{plot_sum_rate_vs_B.eps}
%\caption{Comparisons of the ESRs as a function of $B$ achieved by our joint optimization method, our bits allocation method with equal power allocation, the equal power and bits allocation method and also OMA transmission.}
%\label{fig:plot_sum_rate_vs_B}
%\end{figure}
%%%%%%%%%%%%%%%%%%%%%%%%%%%%%%%%%%%%%%%%%%%%%%%%%%%%%%%%%%%%%%%%%%%%%%%%%%%%%%%%%%%%%%
%%%%%%%%%%%%%%%%%%%%%%%%%%%%%%%%%%%%%%%%%%%%%%%%%%%%%%%%%%%%%%%%%%%%%%%%%%%%%%%%%%%%%

\section{Conclusions}

In this paper, we have investigated the transmit-receive design, performance analysis, and optimization of downlink multi-antenna NOMA cellular networks with general limited feedback in FDD mode.
A mathematically strict performance analysis of the ergodic rate for the considered system was conducted for the first time. A closed-form lower bound on the ergodic rate of each user was obtained in terms of transmit power, CSI quantization accuracy and channel conditions. Then, we optimized the two key parameters of the considered system, i.e., transmit power and the number of feedback bits allocated to each users. Finally, we conducted asymptotic performance analysis on the optimization results, and obtained insights on system design guidelines.

\appendix
\subsection{Proof of Theorem \ref{theorem:R_nk_LB}}\label{proof:theorem_R_nk_LB}

According to \cite[\emph{Lemma} 2]{Jindal06}, we have $\left| \tilde{\mathbf{e}}_{n,k} \mathbf{w}_i \right|^2 \sim \text{Beta} (1, M-2)$ and $\Bbb{E} \left[ | \tilde{\mathbf{e}}_{n,k} \mathbf{w}_i |^2 \right] = (M-1)^{-1}$.  Then, by employing Jensen's inequality we can lower-bound $R_{n,k}$ for $k \geq 2 $ as (\ref{eq:Rnk_1}) at the top of next page,
\begin{figure*}
\begin{eqnarray}\label{eq:R_nk_LB_temp1}
R_{n,k} &\geq & \Bbb{E} \left[  \log_2 \left(  1+   \frac{ \frac{d_{n,k}^{-\alpha}}{\sigma_{n,k}^2} |\mathbf{h}_{n,k} \mathbf{w}_{n}|^2 P_{n,k}} {1 + S_{n,k}^{(2)} \left| \mathbf{h}_{n,k} \mathbf{w}_n \right|^2 + S_{n,k}^{(3)} \Bbb{E} \left[  \left| \tilde{\mathbf{e}}_{n,k} \mathbf{w}_i \right|^2 \right] ||\mathbf{h}_{n,k}||^2 \sin^2\theta_{n,k}}  \right)   \right]  \\
%\end{eqnarray}
%\begin{eqnarray}
\label{eq:Rnk_lower1}
&=& \Bbb{E} \left[ \log_2 \left( \frac{1 + S_{n,k}^{(1)} \left| \mathbf{h}_{n,k} \mathbf{w}_n \right|^2 + S_{n,k}^{(3)}(M-1)^{-1} ||\mathbf{h}_{n,k}||^2 \sin^2\theta_{n,k}}
{1 + S_{n,k}^{(2)} \left| \mathbf{h}_{n,k} \mathbf{w}_n \right|^2 + S_{n,k}^{(3)}(M-1)^{-1} ||\mathbf{h}_{n,k}||^2 \sin^2\theta_{n,k} } \right) \right] \\
&=& \Bbb{E} \left[ \log_2 \left( 1 + S_{n,k}^{(1)} Z ||\mathbf{h}_{n,k}||^2 + S_{n,k}^{(3)}(M-1)^{-1}||\mathbf{h}_{n,k}||^2 \sin^2\theta_{n,k} \right) \right] \nonumber\\
\label{eq:Rnk_1}
&&- \Bbb{E} \left[ \log_2 \left( 1 + S_{n,k}^{(2)} Z ||\mathbf{h}_{n,k}||^2 + S_{n,k}^{(3)}(M-1)^{-1}||\mathbf{h}_{n,k}||^2 \sin^2\theta_{n,k} \right) \right] \triangleq R_{n,k}^{LB1}, ~~
%&\overset{\text{(a)}}{=}& \Bbb{E} \left\{ \log_2 \left[ 1 + S_{n,k}^{(1)} Z ||\mathbf{h}_{n,k}||^2 \cos^2\theta_{n,k} + \left( S_{n,k}^{(1)} Z + S_{n,k}^{(3)}(M-1)^{-1} \right) ||\mathbf{h}_{n,k}||^2 \sin^2\theta_{n,k} \right] \right\} \nonumber\\
%&&- \Bbb{E} \left\{ \log_2 \left[ 1 + S_{n,k}^{(2)} Z ||\mathbf{h}_{n,k}||^2 \cos^2\theta_{n,k} + \left( S_{n,k}^{(2)} Z + S_{n,k}^{(3)}(M-1)^{-1} \right) ||\mathbf{h}_{n,k}||^2 \sin^2\theta_{n,k} \right] \right\} \nonumber\\
\end{eqnarray}
\vspace{-4mm}
\hrulefill
\end{figure*}
where $Z \triangleq | \tilde{\mathbf{h}}_{n,k} \mathbf{w}_n |^2 \sim \text{Beta} (1, M-1)$ \cite{Jindal06}. It has been shown in \cite[\emph{Lemma} 2]{Yoo_limited} that, with cell approximation of vector quantization the joint distribution of the random variable (RV) pair $(||\mathbf{h}_{n,k}||^2 \sin^2\theta_{n,k}, ||\mathbf{h}_{n,k}||^2 \cos^2\theta_{n,k})$ is the same as that of $(I, S)$, where $I = \delta_{n,k} Y$ and $S = X + (1 - \delta_{n,k})Y$ with $X \sim \text{Gamma}(1,1)$ and $Y \sim \text{Gamma}(M-1, 1)$ being two independent RVs.
%(a) is obtained by rewriting $S_{n,k}^{(1)} Z ||\mathbf{h}_{n,k}||^2$ as $S_{n,k}^{(1)} Z ||\mathbf{h}_{n,k}||^2 = S_{n,k}^{(1)} Z ||\mathbf{h}_{n,k}||^2 \cos^2\theta_{n,k} + S_{n,k}^{(1)} Z ||\mathbf{h}_{n,k}||^2 \sin^2\theta_{n,k}$.
%Next, we will calculate the exact value of $R_{n,k,LB}^{(1)}$. Consider two independent Gamma random variables $X \sim \text{Gamma}(1,1)$ and $Y \sim \text{Gamma}(M-1, 1)$, and define $I = \delta_{n,k} Y$ and $S = X + (1 - \delta_{n,k})Y$.
%Though \cite[Lemma 2]{Yoo_limited}, we find that the joint distribution of $(||\mathbf{h}_{n,k}||^2 \sin^2\theta_{n,k}, ||\mathbf{h}_{n,k}||^2 \cos^2\theta_{n,k})$ is the same as that of $(I, S)$.
Since $X$, $Y$ and $Z$ are independent of each other\cite{Jindal06}, $R_{n,k }^{LB1}$ in (\ref{eq:Rnk_1}) can be obtained as
%\begin{eqnarray}\label{eq:Rnk_LB_1}
%R_{n,k }^{LB1} &=& \Bbb{E} \left\{ \log_2 \left[ 1 + S_{n,k}^{(1)} Z ( X + (1 - \delta_{n,k} Y) ) + \left( S_{n,k}^{(1)} Z + S_{n,k}^{(3)}(M-1)^{-1} \right) \delta_{n,k} Y \right] \right\} \nonumber\\
%&&- \Bbb{E} \left\{ \log_2 \left[ 1 + S_{n,k}^{(2)} Z ( X + (1 - \delta_{n,k} Y) ) + \left( S_{n,k}^{(1)} Z + S_{n,k}^{(3)}(M-1)^{-1} \right) \delta_{n,k} Y \right] \right\} \nonumber\\
%&=& \Bbb{E} \left\{ \log_2 \left[ 1 + S_{n,k}^{(1)} Z X  + \left( S_{n,k}^{(1)} Z + S_{n,k}^{(3)}(M-1)^{-1} \delta_{n,k} \right) Y \right] \right\} \nonumber\\
%&&- \Bbb{E} \left\{ \log_2 \left[ 1 + S_{n,k}^{(2)} Z X  + \left( S_{n,k}^{(2)} Z + S_{n,k}^{(3)}(M-1)^{-1} \delta_{n,k} \right) Y \right] \right\} \label{eq:Rnk_2}.
%\end{eqnarray}
\begin{eqnarray}
\hspace{-8mm}&&R_{n,k }^{LB1} = \Bbb{E} \Big\{ \log_2 \Big[ 1 + S_{n,k}^{(1)} Z ( X + (1 - \delta_{n,k} Y) ) \nonumber\\
\hspace{-8mm}&&\hspace{30mm}+ \Big( S_{n,k}^{(1)} Z + S_{n,k}^{(3)}(M-1)^{-1} \Big) \delta_{n,k} Y \Big] \Big\} \nonumber\\
\hspace{-8mm}&&\hspace{13mm}- \Bbb{E} \Big\{ \log_2 \Big[ 1 + S_{n,k}^{(2)} Z ( X + (1 - \delta_{n,k} Y) )  \nonumber\\
\hspace{-8mm}&&\hspace{30mm}+ \Big( S_{n,k}^{(1)} Z + S_{n,k}^{(3)}(M-1)^{-1} \Big) \delta_{n,k} Y \Big] \Big\} \nonumber\\
\hspace{-8mm}&&= \Bbb{E}_{Z} \left[ \Upsilon_{n,k}\left(S_{n,k}^{(1)} Z,~ S_{n,k}^{(1)} Z + S_{n,k}^{(3)}(M-1)^{-1} \delta_{n,k} \right) \right] \nonumber \\
\label{eq:Rnk_LB1_k}
\hspace{-8mm}&&\hspace{2mm}- \Bbb{E}_{Z} \left[ \Upsilon_{n,k}\left(S_{n,k}^{(2)} Z,~ S_{n,k}^{(2)} Z + S_{n,k}^{(3)}(M-1)^{-1} \delta_{n,k} \right) \right] \hspace{-1mm},
\end{eqnarray}
where $ \Upsilon_{n,k}(\mu,\nu) \triangleq \Bbb{E}_{X, Y} \left[ \log_2 \left( 1 + \mu X  + \nu Y \right) \right]$.
The result of $\Bbb{E}_{Z} \left[ \Upsilon_{n,k} ( aZ, aZ+b ) \right] \triangleq \Theta_{n,k} (a,b)$ is derived in \emph{Appendix } \ref{proof:res_E_Upsilon}, where $\Theta_{n,k} (a,b)$ is given in (\ref{eq:res_Theta_nk}). Then, the result of $R_{n,k }^{LB1}$ in (\ref{eq:R_nk_LB1_k}) follows.

%BEGIN---------------- 计算 R_{n,1}^{LB}  --------------------------
When $k = 1$, $S_{n,k}^{(2)} = 0$. By replacing $S_{n,k}^{(2)}$ in (\ref{eq:Rnk_LB1_k}) with 0, we can similarly obtain
\begin{eqnarray}
\label{eq:Rnk_LB1_1}
R_{n,1}^{LB1}  &=& \Bbb{E} \Big\{ \log_2 \Big[ 1 + \alpha_{n,1} P_{n,1} Z X  \nonumber\\
&&\hspace{3mm}+ \left( \alpha_{n,1} P_{n,1} Z + S_{n,1}^{(3)}(M-1)^{-1} \delta_{n,1} \right) Y \Big] \Big\} \nonumber\\
&&- \Bbb{E} \left\{ \log_2 \left[ 1 + S_{n,1}^{(3)}(M-1)^{-1} \delta_{n,1} Y \right] \right\} \nonumber\\
&=& \Theta_{n,k}(\alpha_{n,1} P_{n,1}, S_{n,1}^{(3)}(M-1)^{-1} \delta_{n,1}) \nonumber\\
&&- \Bbb{E}  \left[ \log_2 \left( 1 + S_{n,1}^{(3)}(M-1)^{-1} \delta_{n,1} Y \right) \right]\label{eq:Rn1_mid}.~~~
\end{eqnarray}
The second term in (\ref{eq:Rn1_mid}) can be obtained as
\begin{eqnarray}
\label{eq:Rnk_LB1_mid}
\hspace{-8mm}&&\Bbb{E} \left[ \log_2 \left( 1 + S_{n,1}^{(3)}(M-1)^{-1} \delta_{n,1} Y \right) \right] \nonumber\\
\hspace{-8mm}&&=  \int_0^{+\infty} \log_2 \left( 1 + S_{n,1}^{(3)}(M-1)^{-1} \delta_{n,1} y \right) \frac{e^{-y} y^{M-2}}{(M-2)!} ~dy \nonumber\\
%&=& \frac{\log_2(e)}{(M-2)!} \int_0^{+\infty} \ln \left( 1 + S_{n,1}^{(3)}(M-1)^{-1} \delta_{n,1} y \right) e^{-y} y^{M-2} ~dy \nonumber\\
\hspace{-8mm}&&= \log_2(e) \exp \left( \frac{M-1}{S_{n,1}^{(3)} \delta_{n,1}} \right) \sum_{q=1}^{M-1} E_q \left( \frac{M-1}{S_{n,1}^{(3)} \delta_{n,1}} \right),
\end{eqnarray}
where (\ref{eq:Rnk_LB1_mid}) follows by using (\ref{eq:int_lnex}) in \emph{Appendix} \ref{proof:res_E_Upsilon}.
Then, $R_{n,k }^{LB1}$ in (\ref{eq:R_nk_LB1_1}) for $k =1$ follows by substituting (\ref{eq:Rnk_LB1_mid}) into (\ref{eq:Rnk_LB1_1}).
%END---------------- 计算 R_{n,1}^{LB}  --------------------------

\subsection{The derivation of $\Bbb{E}_{Z}\left[ \Upsilon_{n,k} ( aZ, aZ+b ) \right]  \triangleq \Theta_{n,k} (a,b)  $ for $a>0$, $b>0$}\label{proof:res_E_Upsilon}
%\begin{lemma}
%The result of $\Bbb{E} \left[ \Upsilon_{n,k} ( aZ, aZ+b ) \right]$ is
%\begin{eqnarray}\label{eq:Upsilon_ab}
%\Bbb{E} \left[ \Upsilon_{n,k} ( aZ, aZ+b ) \right] &=& \Theta_{n,k} (a,b),
%\end{eqnarray}
%where $\Theta_{n,k} (a,b)$ is given by (\ref{eq:res_Theta_nk}) and the definition of $\Upsilon_{n,k} (\mu, \nu)$ is given in the proof of \emph{Theorem \ref{theorem:R_nk_LB}}.
%\end{lemma}
%First, we compute the result of $\Upsilon_{n,k}(\mu,\nu)$.
Using the result in \cite{Amari97}, the probability density function (PDF) of $J_{\mu, \nu} = \mu X  + \nu Y$ for $\mu \neq \nu$ is given by $f_{J_{\mu, \nu}}(x) = \frac{(-1)^{M-1} \mu^{M-2}}{(\nu - \mu)^{M-1}} e^{-\frac{x}{\mu}} + \sum_{p=1}^{M-1} \frac{(-\mu)^{p-1}}{(M-1-p)! (\nu-\mu)^p \nu^{M-1-p}} x^{M-1-p} e^{-\frac{x}{\nu}}~{(x > 0)}$. Thus,
\begin{eqnarray}
&&\Upsilon_{n,k}(\mu,\nu) =  \Bbb{E}_{X,Y} \left[ \log_2 \left( 1 + \mu X  + \nu Y \right) \right]\nonumber\\ &&= \int_0^{+\infty} \log_2(1 + x) f_{J_{\mu, \nu}}(x)~dx \nonumber\\
\label{eq:Upsilon_nk}
&&= \log_2(e) \Bigg[ \frac{(-\mu)^{M-1}}{(\nu - \mu)^{M-1}} \exp\left(\frac{1}{\mu}\right) E_1\left(\frac{1}{\mu}\right) \nonumber\\
&&\hspace{4mm}+ \sum_{p=1}^{M-1} \frac{(-\mu)^{p-1} \nu}{(\nu - \mu)^p} \exp \left( \frac{1}{\nu} \right) \sum_{q=1}^{M-p} E_q \left( \frac{1}{\nu} \right) \Bigg], ~~~~~~
\end{eqnarray}
where (\ref{eq:Upsilon_nk}) follows from the result of \cite[(78)]{Alouini99}
\begin{eqnarray}\label{eq:int_lnex}
\int_0^{+\infty} \ln(1+ax) e^{-\mu x} x^{p-1}dx \nonumber\\= (p-1)! e^{\frac{\mu}{a}} \mu^{-p} \sum_{q=1}^{p} E_q\left(\frac{\mu}{a}\right).
\end{eqnarray}
Then, substituting $\mu = aZ$ and $\nu = aZ + b$ with $a,b >0$ in to (\ref{eq:Upsilon_nk}), we can obtain
%BEGIN------------------ 计算 \Bbb{E} \left[ \Upsilon_{n,k} ( aZ, aZ+b ) \right] -----------------------
\begin{eqnarray}\label{eq:E_Upsilon_1}
\hspace{-9mm}&&\Bbb{E}_{Z} \left[ \Upsilon_{n,k} ( aZ, aZ+b ) \right] \nonumber\\
\hspace{-9mm}&&=  \log_2(e) \Bbb{E}_{Z} \Bigg[
\frac{(-aZ)^{M-1}}{b^{M-1}} \exp \bigg( \frac{1}{aZ} \bigg) E_1 \bigg( \frac{1}{aZ} \bigg) \nonumber\\
\hspace{-9mm}&&\hspace{22mm} + \sum_{p=1}^{M-1} \frac{(-aZ)^{p-1} (aZ+b)}{b^p} \exp \bigg( \frac{1}{aZ+b} \bigg) \nonumber\\
\hspace{-9mm}&&\hspace{22mm} \times \sum_{q=1}^{M-p} E_q \bigg( \frac{1}{aZ+b} \bigg)\Bigg].
\end{eqnarray}
By using the relationship\cite{Gradshteyn2007} $E_{q}(x) = \frac{1}{q-1} \left[ e^{-x} - x E_{q-1}(x) \right]$ and $E_1(x) = -E_i(-x)$ with induction, $E_{q}(x)$ can be expressed with $E_i(-x)$ as $E_{q}(x) =  e^{-x} \sum_{s=0}^{q-2} \frac{(q-s-2)!(-x)^s}{(q-1)!} + \frac{(-1)^q (x)^{q-1}}{(q-1)!}E_{i}(-x)$. Then, substituting this result into (\ref{eq:E_Upsilon_1}) with the PDF of $Z$, we can obtain
\vspace{-1mm}
\begin{eqnarray}
\label{eq:E_Upsilon_2}
\hspace{-0.5cm}&&\Bbb{E}_{Z} \left[ \Upsilon_{n,k} ( aZ, aZ+b ) \right] \nonumber\\
\hspace{-0.5cm}&=& \log_2(e) \Bbb{E}_{Z} \Bigg[
\frac{(-1)^M(aZ)^{M-1}}{b^{M-1}} \exp \left( \frac{1}{aZ} \right) E_i \left( -\frac{1}{aZ} \right) \nonumber\\
\hspace{-0.5cm}&& + \sum_{p=1}^{M-1} \sum_{q=1}^{M-p} \frac{ (-1)^{p+q+1} (aZ)^{p-1} }{ (q-1)! b^p (aZ+b)^{q-2} } \nonumber\\
\hspace{-0.5cm}&& \hspace{2cm} \times \exp \left( \frac{1}{aZ+b} \right)  E_i \left( -\frac{1}{aZ+b} \right) \nonumber\\
\hspace{-0.5cm}\hspace{-0.8cm}&& + \sum_{p=1}^{M-1} \sum_{q=2}^{M-p} \sum_{s=0}^{q-2} \frac{ (-1)^{p+s-1} (q-s-2)! (aZ)^{p-1} }{ (q-1)! b^{p} (aZ+b)^{s-1} }\Bigg],
\end{eqnarray}
which is given by (\ref{eq:res_Theta_nk}) with
\begin{eqnarray}
\hspace{-7mm}&&I_1(a) \triangleq \int_0^1 (az)^{M-1} (1-z)^{M-2} \exp \bigg( \frac{1}{az} \bigg) E_i \bigg( \hspace{-1.5mm}-\frac{1}{az} \bigg) dz,\nonumber\\
\hspace{-7mm}&&\label{I1_definition} \\
\label{I2_definition}
\hspace{-7mm}&&I_2(a,b,p,q) \triangleq \int_0^1 \frac{ (az)^{p-1} (1-z)^{M-2} }{ (az+b)^{q-2} } \exp \left( \frac{1}{az+b} \right) \nonumber\\
\hspace{-7mm}&&\hspace{40mm}  \times E_i \left( -\frac{1}{az+b} \right) dz, \\
\label{I3_definition}
\hspace{-7mm}&&I_3(a,b,p,s) \triangleq \int_0^1 \frac{ (az)^{p-1} (1-z)^{M-2} }{ (az+b)^{s-1} } dz.
\end{eqnarray}
Then, we derive the results of (\ref{I1_definition}), (\ref{I2_definition}) and (\ref{I3_definition}) as follows. For the convenience of presentation, we first define an integral $\Psi(n, u, v) = \int_u^v x^n e^x E_i(-x) dx$, whose closed-form result is given by (\ref {eq:res_Psi_1}) for different cases. The derivation is provided in \emph{Appendix} \ref{appendix:Psi}.
Then, $I_1(a)$ can be obtained as
\begin{eqnarray}
\hspace{-0.7cm}&&I_1(a) = \sum_{t=0}^{M-2} {M-2 \choose t} \frac{(-1)^t}{a^t} \int_0^1 (az)^{M-1+t} \nonumber\\
\hspace{-0.7cm}&&\hspace{4cm} \times \exp \left( \frac{1}{az} \right) E_i \left( -\frac{1}{az} \right) dz \nonumber\\
\hspace{-0.7cm}&&= \sum_{t=0}^{M-2} {M-2 \choose t} \frac{(-1)^t}{a^t} \int_{+\infty}^{\frac{1}{a}} x^{-M+1-t} e^x E_i(-x) ~d\left( \frac{1}{ax} \right) \nonumber\\
\hspace{-0.7cm}&&\quad\label{eq:res_I1}\\
\hspace{-0.7cm}&&= \sum_{t=0}^{M-2} {M-2 \choose t} \frac{(-1)^t}{a^{t+1}} \Psi\left(-M-1-t, \frac{1}{a}, v\right)\bigg|_{v \to +\infty}\nonumber,
\end{eqnarray}
where (\ref{eq:res_I1}) follows by the change of variables $x = \frac{1}{az}$.
Similarly, $I_2(a,b,p,q)$ can be obtained by using the change of variables $x = \frac{1}{az + b}$ as
\begin{eqnarray}
\hspace{-7mm}&& I_2(a,b,p,q)  =\int_{\frac{1}{b}}^{\frac{1}{a+b}} \left( \frac{1}{x} - b \right)^{p-1} \left[ 1 - \frac{1}{a} \left( \frac{1}{x} - b \right) \right]^{M-2} \nonumber\\
\hspace{-7mm}&&\hspace{30mm} \times x^{q-2} e^x E_i(-x) ~d \left[ \frac{1}{a} \left( \frac{1}{x} - b \right) \right] \nonumber\\
\hspace{-7mm}&&= \sum_{r=0}^{M-2} \sum_{t=0}^{p-1+r} {M-2 \choose r} {p-1+r \choose t} \frac{(-1)^{p-1-t} b^{p-1+r-t}}{a^{r+1}} \nonumber\\
\hspace{-7mm}&&\hspace{40mm} \times\int_{\frac{1}{a+b}}^{\frac{1}{b}} x^{q-4-t} e^x E_i(-x) ~dx \nonumber\\
\hspace{-7mm}&&= \sum_{r=0}^{M-2} \sum_{t=0}^{p-1+r} {M-2 \choose r} {p-1+r \choose t} \frac{(-1)^{p-1-t} b^{p-1+r-t}}{a^{r+1}} \nonumber\\
\hspace{-7mm}&&\hspace{40mm} \times\Psi\left(q-4-t, \frac{1}{a+b}, \frac{1}{b}\right).\nonumber
\end{eqnarray}
Finally, $I_3(a,b,p,s)$ can be obtained as
%\vspace{-1mm}
\begin{eqnarray}
\hspace{-1.3cm}&&I_3(a,b,p,s)  \nonumber\\
\hspace{-1.3cm}&&= \sum_{t=0}^{M-2} {M-2 \choose t} (-a)^{-t} \int_0^1 \frac{ (az)^{p-1+t} }{ (az+b)^{s-1} } dz \nonumber\\
\hspace{-1.3cm}&&= \sum_{t=0}^{M-2} {M-2 \choose t} (-1)^{t} a^{-t-1} \int_0^a \frac{ x^{p-1+t} }{ (x+b)^{s-1} } dx \nonumber\\
\label{eq:res_I3_1}
\hspace{-1.3cm}&&= \sum_{t=0}^{M-2} {M-2 \choose t} \frac{ (-1)^{t} a^{-t-1} }{ b^{s-1} } \int_0^a \frac{ x^{p-1+t} }{ (\frac{x}{b}+1)^{s-1} } dx\\
\label{eq:res_I3_2}
\hspace{-1.3cm}&&= \sum_{t=0}^{M-2} {M-2 \choose t} \frac{ (-1)^{t} a^{-t-1} }{ b^{s-1} } \cdot \frac{ a^{p+t} }{ (p+t) } \nonumber\\
\hspace{-1.3cm}&&\hspace{2cm} \times ~{_2}F_1 \left( s-1, p+t; p+t+1; -\frac{a}{b} \right)~~~~ \\
\hspace{-1.3cm}&&= \sum_{t=0}^{M-2} {M-2 \choose t} \frac{ (-1)^{t} a^{p-1} }{ b^{s-1} (p+t) } \nonumber\\
\hspace{-1.3cm}&&\hspace{2cm} \times ~{_2}F_1 \left( s-1, p+t; p+t+1; -\frac{a}{b} \right) , \nonumber
\end{eqnarray}
where (\ref{eq:res_I3_1}) is obtained from the change of variables $x = az$ and (\ref{eq:res_I3_2}) follows from \cite[3.194.1]{Gradshteyn2007} $\int_0^u \frac{x^{\mu-1} dx}{(1+\beta x)^{\nu}} = \frac{u^{\mu}}{\mu} {_2}F_1 (\nu, \mu; 1+ \mu; -\beta u)$.
Then, the final results can be obtained by substituting the results of $I_1(a)$, $I_2(a,b,p,q)$ and $I_3(a,b,p,s)$ into (\ref{eq:E_Upsilon_2}).

\subsection{Derivation of $\Psi(n, u, v)$}\label{appendix:Psi}
First, for $n>0$, $\Psi(n, u, v)$ can be obtained by integration by parts as
\begin{eqnarray}\label{eq:result1}
\hspace{-5mm}\Psi(n, u, v) &=&  \int_u^v E_i(-x) ~d \left[ e^x \sum_{k=0}^n \frac{(-1)^k n!}{(n-k)!} x^{n-k} \right]  \\
\hspace{-5mm}&=& \left[ E_i(-x) e^x \sum_{k=0}^n \frac{(-1)^k n!}{(n-k)!} x^{n-k} \right]_u^v \nonumber\\
\hspace{-5mm}&&- \int_u^v e^x \sum_{k=0}^n \frac{(-1)^k n!}{(n-k)!} x^{n-k} ~d [E_i(-x)] \nonumber\\
\label{eq:result2}
\hspace{-5mm}&=& \left[ E_i(-x) e^x \sum_{k=0}^n \frac{(-1)^k n!}{(n-k)!} x^{n-k} \right]_u^v \nonumber\\
\hspace{-5mm}&&- \sum_{k=0}^n \frac{(-1)^k n!}{(n-k)!} \int_u^v x^{n-k-1} ~dx,
\end{eqnarray}
which is given in (\ref{eq:res_Psi_1}) with $n > 0$, where (\ref{eq:result1}) follows from the result \cite[2.321.2]{Gradshteyn2007} $\int x^n e^{ax} dx = e^{ax}   \Big( \sum_{k=0}^{n} \frac{(-1)^k n!}{a^{k+1} (n-k)!} x^{n-k} \Big)$ and (\ref{eq:result2}) follows from $\left( E_i(-x) \right)' = x^{-1} e^{-x}$.

For $n=0$, we can similarly obtain
\begin{eqnarray}
\hspace{-9mm}&&\Psi(0, u, v)=  \int_u^v E_i(-x) {\rm d} e^x  =  e^x E_i(-x)\Big|_{u}^{v} - \int_u^v x^{-1} {\rm d} x \nonumber\\
\hspace{-9mm}&&\hspace{15.5mm}=  e^v E_i(-v) - e^u E_i(-u) - \ln \left( \frac{v}{u} \right). \nonumber
\end{eqnarray}

For $n\leq -2 $, $\Psi(n, u, v)$ can be obtained as
\begin{eqnarray}\label{eq:result3}
\hspace{-8mm}&&\Psi(n, u, v)
= \int_u^v E_i(-x) d \Bigg[\hspace{-1mm} -e^x \sum_{k=1}^{-n-1} \frac{(-n-k-1)!}{(-n-1)!} x^{n+k} \nonumber\\
\hspace{-7mm}&&\hspace{55mm}+ \frac{E_i(x)}{(-n-1)!} \Bigg] \\
\hspace{-7mm}&&= \Bigg[ E_i(-x) \Bigg( -e^x \sum_{k=1}^{-n-1} \frac{(-n-k-1)!}{(-n-1)!} x^{n+k} \nonumber\\
\hspace{-7mm}&&\hspace{55mm}+ \frac{E_i(x)}{(-n-1)!} \Bigg) \Bigg]_u^v \nonumber\\
\label{eq:result4}
\hspace{-7mm}&&- \int_u^v \Bigg( -e^x \sum_{k=1}^{-n-1} \frac{(-n-k-1)!}{(-n-1)!} x^{n+k} \nonumber\\
\hspace{-7mm}&&\hspace{40mm}+ \frac{E_i(x)}{(-n-1)!} \Bigg) ~d [E_i(-x)],
\end{eqnarray}
where (\ref{eq:result3}) follows from the result \cite[2.324.2]{Gradshteyn2007} $\int \frac{e^{ax}}{x^n} dx = -e^{ax} \sum_{k=1}^{n-1} \frac{(n-k-1)! a^{k-1}}{(n-1)! x^{n-k}} + \frac{a^{n-1} E_i(ax)}{(n-1)!}  (n>0)$ with $a=1$. Moreover, $\int_u^v E_i(x) ~d[E_i(-x)]$ can be obtained as
\begin{eqnarray}
\label{eq:result5}
\hspace{-8mm}&& \int_u^v E_i(x) ~d[E_i(-x)]= \int_u^v \bigg( E_i(-x) \nonumber\\
\hspace{-8mm}&&\hspace{15mm}+\sum_{m=0}^{\infty} \frac{2 x^{2m+1}}{(2m+1)(2m+1)!}\bigg) ~d[E_i(-x)]\\
\label{eq:result5_2}
\hspace{-8mm}&&= \frac{1}{2} E_i^2(-x) \Big|_u^v+ \sum_{m=0}^{\infty} \frac{ 2 }{ (2m+1)(2m+1)! } \nonumber\\
\hspace{-8mm}&&\hspace{15mm}\times\left( \int_0^v x^{2m} e^{-x} ~dx - \int_0^u x^{2m} e^{-x} ~dx \right) \nonumber\\
\label{eq:result6}
\hspace{-8mm}&&= \frac{1}{2} \left[ E_i^2(-v) - E_i^2(-u) \right]
+ \sum_{m=0}^{\infty} \sum_{l=0}^{2m} \frac{2 (u^l e^{-u} - v^l e^{-v})}{(2m+1)^2\,l!}, \nonumber\\
\hspace{-8mm}&&
\end{eqnarray}
where (\ref{eq:result5}) follows from\cite[8.214.3]{Gradshteyn2007} $E_i(x) = E_i(-x) + \sum_{m=0}^{\infty} \frac{2 x^{2m+1}}{(2m+1)(2m+1)!}$ and (\ref{eq:result6}) follows from \cite[3.351.1]{Gradshteyn2007}. Then, the expression of $\Psi(n, u, v)$ for $n\leq -2$ in (\ref{eq:res_Psi_1}) can be obtained by combing (\ref{eq:result4}) and (\ref{eq:result6}).
Similarly, for $n = -1$, we can obtain
\begin{eqnarray}
\hspace{-0.7cm}&& \Psi(-1, u, v) = \int_u^v x^{-1} e^x E_i(-x) ~dx = \int_u^v E_i(-x) ~d E_i(x) \nonumber\\
\hspace{-0.7cm}&&\hspace{1.8cm}= \left[ E_i(-x) E_i(x) \right]_u^v -\hspace{-1mm} \int_u^v E_i(x) ~d [E_i(-x)],
%\hspace{-1cm}  & = &E_i(-v)E_i(v) - E_i(-u)E_i(u) + \frac{E_i^2(-u) - E_i^2(-v)}{2} + \frac{1}{2}\sum_{m=0}^{\infty} \sum_{l=0}^{2m} \frac{2 (v^l e^{-v} - u^l e^{-u})}{(2m+1)^2\,l!}. \nonumber
\end{eqnarray}
which is given by (\ref{eq:res_Psi_1}) with $n= -1$. Here, we employ again \cite[3.351.1]{Gradshteyn2007}.

For $n<-1$ with $v\to+\infty$, we can obtain
\begin{eqnarray}\label{eq:Psi_nless0_inf_1}
\hspace{-7mm}&&\tilde{\Psi}(n, u) = \Psi(n, u, v)\big|_{v \to +\infty} = \int_u^{\infty} x^n e^x E_i(-x) ~dx \nonumber\\
\hspace{-7mm}&&= \Bigg[ E_i(-x) \Bigg( -e^x \sum_{k=1}^{-n-1} \frac{(-n-k-1)!}{(-n-1)!} x^{n+k} \nonumber\\
\hspace{-7mm}&&\hspace{4mm}+ \frac{E_i(x)}{(-n-1)!} \Bigg) \Bigg]_u^{\infty}
+ \left[ \sum_{k=1}^{-n-1} \frac{ (-n-k-1)!~x^{n+k} }{ (-n-1)!(n+k) } \right]_u^{\infty} \nonumber\\
\hspace{-7mm}&&\hspace{4mm}- \frac{1}{(-n-1)!} \int_u^{\infty} E_i(x) ~d [E_i(-x)].
\end{eqnarray}
%For the first term in (\ref{eq:Psi_nless0_inf_1}), we need the results of $\lim \limits_{x \to +\infty} e^x E_i(-x)$ and $\lim \limits_{x \to +\infty} E_i(-x) E_i(x)$.
According to \cite[5.1.19]{abramowitz1965handbook}, we have $\frac{1}{x+1} < e^x E_1(x) \leq \frac{1}{x}~(x>0)$. Since $E_1(x) = -E_i(-x)$, we have $-\frac{1}{x} < e^x E_i(-x) \leq -\frac{1}{x+1}~(x>0)$.
Then, it is easy to see $\lim \limits_{x \to +\infty} e^x E_i(-x) = 0$. In addition, using L{'}Hospital rule it is easy to obtain $\lim \limits_{x \to +\infty} E_i(-x) E_i(x) = 0$.
%\begin{eqnarray}
%\lim \limits_{x \to +\infty} E_i(-x) E_i(x)
%&=& \lim \limits_{x \to +\infty} \frac{E_i(x)}{E_i^{-1}(-x)}
%= \lim \limits_{x \to +\infty} \frac{x^{-1}e^x}{-E_i^{-2}(-x)x^{-1}e^{-x}}\nonumber\\
%&=& -\lim \limits_{x \to +\infty} \left[ e^x E_i(-x) \right]^2 = 0.
%\end{eqnarray}
Thus, for $n+k < 0$, the first term in (\ref{eq:Psi_nless0_inf_1}) can be obtained as $\left[ E_i(-x) \left( -e^x \sum_{k=1}^{-n-1} \frac{(-n-k-1)!}{(-n-1)!} x^{n+k} + \frac{E_i(x)}{(-n-1)!} \right) \right]_u^{\infty}  =  E_i(-u) \left( e^u \sum_{k=1}^{-n-1} \frac{(-n-k-1)!}{(-n-1)!} u^{n+k} - \frac{E_i(u)}{(-n-1)!} \right) $.
%\begin{eqnarray}
%&&\left[ E_i(-x) \left( -e^x \sum_{k=1}^{-n-1} \frac{(-n-k-1)!}{(-n-1)!} x^{n+k} + \frac{E_i(x)}{(-n-1)!} \right) \right]_u^{\infty} \nonumber\\
%&=& E_i(-u) \left( e^u \sum_{k=1}^{-n-1} \frac{(-n-k-1)!}{(-n-1)!} u^{n+k} - \frac{E_i(u)}{(-n-1)!} \right).
%\end{eqnarray}
The second term in (\ref{eq:Psi_nless0_inf_1}) is given by $\bigg[ \sum_{k=1}^{-n-1} \frac{ (-n-k-1)!~x^{n+k} }{ (-n-1)!(n+k) } \bigg]_u^{\infty}   = -\sum_{k=1}^{-n-1} \frac{ (-n-k-1)!~u^{n+k} }{ (-n-1)!(n+k) }$. Similar to (\ref{eq:result6}), the third term in (\ref{eq:Psi_nless0_inf_1}) can be obtained as $\int_u^{+\infty} E_i(x) ~d[E_i(-x)] = -\frac{1}{2} E_i^2(-u) + \sum_{m=0}^{\infty} \sum_{l=0}^{2m} \frac{2 u^l e^{-u}}{(2m+1)^2\,l!}$,
%\begin{eqnarray}
%&&\int_u^{+\infty} E_i(x) ~d[E_i(-x)] \nonumber\\
%&=& \int_u^{+\infty} E_i( - x) ~d[E_i(-x)]
%+ \sum_{m=0}^{\infty} \frac{ 2 }{ (2m+1)(2m+1)! } \int_u^{+\infty} x^{2m} e^{-x} ~dx \nonumber\\
%\label{eq:result7}
%&{=}& -\frac{1}{2} E_i^2(-u)
%+ \sum_{m=0}^{\infty} \sum_{l=0}^{2m} \frac{2 (2m)! u^l e^{-u}}{(2m+1)(2m+1)!\,l!} \nonumber\\
%&=& -\frac{1}{2} E_i^2(-u)
%+ \sum_{m=0}^{\infty} \sum_{l=0}^{2m} \frac{2 u^l e^{-u}}{(2m+1)^2\,l!}
%\end{eqnarray}
where we have used \cite[3.351.2]{Gradshteyn2007}. Then, the result of $\Psi(n, u, v)\big|_{v \to +\infty}$ in (\ref{eq:res_Psi_2}) follows by substituting the above results into (\ref{eq:Psi_nless0_inf_1}).

\subsection{Proof of Theorem \ref{theorem:R_nk_Loss_Upper}}\label{proof:theorem_R_nk_Loss_Upper}
With the expressions of $R_{n,k}^{ideal} $ and $ R_{n,k} $ we have
\begin{small}
\begin{eqnarray}
\hspace{-9mm}&&\Delta R_{n,k}
= \Bbb{E} \left[ \log_2 \left( 1 + S_{n,k}^{(1)} \left| \mathbf{h}_{n,k} \mathbf{w}_n \right|^2 \right) \right]\hspace{-5mm}\nonumber\\
\hspace{-9mm}&&- \Bbb{E} \left[ \log_2 \left( 1 + S_{n,k}^{(2)} \left| \mathbf{h}_{n,k} \mathbf{w}_n \right|^2 \right) \right] \hspace{-1mm} - \hspace{-0.6mm} \Bbb{E} \Bigg[ \log_2 \bigg( 1 + S_{n,k}^{(1)} \left| \mathbf{h}_{n,k} \mathbf{w}_n \right|^2 \hspace{-5mm}\nonumber\\
\hspace{-9mm}&&\hspace{10mm}+ \frac{d_{n,k}^{-\alpha}}{\sigma_{n,k}^2} \sin^2\theta_{n,k}||\mathbf{h}_{n,k}||^2 \sum_{i=1,i \neq n}^{N} \left| \tilde{\mathbf{e}}_{n,k} \mathbf{w}_i \right|^2 \sum_{l=1}^{K} P_{i,l} \bigg) \Bigg] \hspace{-5mm}\nonumber\\
\hspace{-9mm}&&\hspace{2mm}+~ \Bbb{E} \Bigg[ \log_2 \Bigg( 1 + S_{n,k}^{(2)} \left| \mathbf{h}_{n,k} \mathbf{w}_n \right|^2 + \frac{d_{n,k}^{-\alpha}}{\sigma_{n,k}^2} \sin^2\theta_{n,k}||\mathbf{h}_{n,k}||^2 \hspace{-5mm}\nonumber\\
\label{eq:Delta_Rnk_1}
\hspace{-9mm}&&\hspace{25mm}\times \sum_{i=1,i \neq n}^{N} \left| \tilde{\mathbf{e}}_{n,k} \mathbf{w}_i \right|^2 \sum_{l=1}^{K} P_{i,l} \Bigg) \Bigg] \hspace{-5mm} \\
\label{eq:Delta_Rnk_2}
\hspace{-9mm}&&\leq \Bbb{E} \Bigg[ \log_2 \Bigg( 1 + S_{n,k}^{(2)} \left| \mathbf{h}_{n,k} \mathbf{w}_n \right|^2 + \frac{d_{n,k}^{-\alpha}}{\sigma_{n,k}^2} \sin^2\theta_{n,k} ||\mathbf{h}_{n,k}||^2 \hspace{-5mm}\\
\hspace{-9mm}&&\hspace{1mm}\times \hspace{-1mm} \sum_{i=1,i \neq n}^{N} \left| \tilde{\mathbf{e}}_{n,k} \mathbf{w}_i \right|^2 \sum_{l=1}^{K} P_{i,l} \Bigg) \Bigg]
\hspace{-1mm}-\hspace{-0.6mm} \Bbb{E} \left[ \log_2 \left( 1 + S_{n,k}^{(2)} \left| \mathbf{h}_{n,k} \mathbf{w}_n \right|^2 \right] \right) \hspace{-5mm}\nonumber\\
\label{eq:Delta_Rnk_3}
\hspace{-9mm}&&\leq \log_2 \left( 1 + S_{n,k}^{(2)} + \frac{M}{M-1} 2^{B_{n,k}} \beta \left( 2^{B_{n,k}}, \frac{M}{M-1} \right) S_{n,k}^{(3)} \right) \hspace{-5mm}\nonumber\\
\hspace{-9mm}&&\hspace{8mm}- \log_2(e) e^{\frac{1}{S_{n,k}^{(2)}}} E_1 \left( \small{\frac{1}{S_{n,k}^{(2)}}} \right) \hspace{-5mm}\nonumber\\
\label{eq:Delta_R_nk_4}
\hspace{-9mm}&&= \log_2 \left( 1 + S_{n,k}^{(2)} + \Gamma \left( \frac{2M-1}{M-1} \right) 2^{-\frac{B_{n,k}}{M-1}}  S_{n,k}^{(3)} \right) \hspace{-5mm}\nonumber\\
\hspace{-9mm}&&\hspace{8mm}- \log_2(e) e^{\frac{1}{S_{n,k}^{(2)}}} E_1 \left( \small{\frac{1}{S_{n,k}^{(2)}}} \right),
\end{eqnarray}
\end{small}
where (\ref{eq:Delta_Rnk_2}) follows by neglecting $\frac{d_{n,k}^{-\alpha}}{\sigma_{n,k}^2} \sin^2\theta_{n,k}||\mathbf{h}_{n,k}||^2 \sum_{i=1,i \neq n}^{N} \left| \tilde{\mathbf{e}}_{n,k} \mathbf{w}_i \right|^2 \sum_{l=1}^{K} P_{i,l}$ in the third $\log_2$ of (\ref{eq:Delta_Rnk_1}).
%This positive term is quite small when the feedback bits $B_{n,k}$ is relatively large, and then the first term of (\ref{eq:Delta_Rnk_1}) will be equal to the third term of (\ref{eq:Delta_Rnk_1}).
(\ref{eq:Delta_Rnk_3}) is obtained by applying Jensen’s inequality to the first $\log_2$ term of (\ref{eq:Delta_Rnk_2}).
(\ref{eq:Delta_R_nk_4}) follows by combing the results $\Bbb{E} [| \mathbf{h}_{n,k} \mathbf{w}_n |^2] = 1$, $\Bbb{E} [||\mathbf{h}_{n,k}||^2] = M$ and $\Bbb{E} [\sin^2\theta_{n,k}] = 2^{B_{n,k}} \beta \Big( 2^{B_{n,k}}, \frac{M}{M-1} \Big)$ with $2^{B_{n,k}} \beta \left( 2^{B_{n,k}}, \frac{M}{M-1} \right) \approx \Gamma \left( \frac{M}{M-1} \right) 2^{-\frac{B_{n,k}}{M-1}}$ and $x \Gamma(x) = \Gamma(x+1)$, which were shown in \cite{Jindal06}.

\subsection{Proof of Theorem \ref{theorem:B_without_integer}} \label{proof:theorem_B_without_integer}
%\begin{eqnarray}\label{eq:bit_prod}
%\min_{B_{n,k}} &&\prod_{n=1}^N \prod_{k=1}^{K} \left( 1 + S_{n,k}^{(2)} + \Gamma \left( \frac{2M-1}{M-1} \right) 2^{-\frac{B_{n,k}}{M-1}}  S_{n,k}^{(3)} \right) \\
%\text{s.t.} && \sum_{n=1}^{N} \sum_{k=1}^{K} B_{n,k} \leq B, \nonumber\\
%&& B_{n,k} \geq 0. \nonumber
%\end{eqnarray}
It is easy to show that the objective function is logarithm convex function of variables $B_{n, k}$s. Thus, the problem of (\ref{problem:bit_allocation_prod}) is a convex optimization problem. The Lagrangian is given by
\begin{eqnarray}
\hspace{-8mm}&&\mathcal{L}(B_{n,k}, \lambda) = \lambda \left( \sum_{n=1}^{N} \sum_{k=1}^{K} B_{n,k} - B \right) \nonumber\\
\hspace{-8mm}&&\hspace{8mm}+ \prod_{n=1}^N \prod_{k=1}^{K} \left( 1 + S_{n,k}^{(2)} + \Gamma \left( \frac{2M-1}{M-1} \right) 2^{-\frac{B_{n,k}}{M-1}}  S_{n,k}^{(3)} \right), \nonumber
\end{eqnarray}
where $\lambda$ is the is the Lagrange multiplier associated with total feedback bits constraint. Applying the Karush-Kuhn-Tucker (KKT) condition, we can get the necessary and sufficient conditions as
\vspace{-1mm}
\begin{eqnarray}\label{eq:KKT_Bnk}
\hspace{-8mm}&& \frac{\partial \mathcal{L} }{\partial B_{n,k}} \Bigg|_{ B^{\star}_{n,k}, ~\lambda^{\star}} = -\frac{\ln(2) \Gamma\left(\frac{2M-1}{M-1}\right) S_{n,k}^{(3)} 2^{-\frac{B_{n,k}}{M-1}}}{(M-1)} \nonumber\\
\hspace{-8mm}&&\times \prod_{i=1}^N \prod_{j=1}^{K} \Bigg( 1 + S_{i,j}^{(2)}  + \Gamma \left( \frac{2M-1}{M-1} \right) 2^{-\frac{B_{i,j}^{\star}}{M-1}} S_{i,j	}^{(3)} \Bigg) \nonumber\\
\hspace{-8mm}&&\times  \left( 1 + S_{n,k}^{(2)} + \Gamma \left( \frac{2M-1}{M-1} \right) 2^{-\frac{B_{n,k}^{\star}}{M-1}} S_{n,k}^{(3)} \right)^{-1} + \lambda^{\star} = 0, \nonumber\\
\hspace{-8mm}&&\hspace{65mm}\forall ~ n,k,
\end{eqnarray}
\begin{eqnarray}\label{eq:KKT_lambda}
\frac{\partial \mathcal{L} (B_{n,k}, \lambda)}{\partial \lambda} \Bigg|_{B^{\star}_{n,k},~ \lambda^{\star}} =
\sum_{n=1}^{N} \sum_{k=1}^{K} B_{n,k}^{\star} - B = 0,&&  \nonumber\\
\text{if}~~  \lambda^{\star} \geq 0.&&
\end{eqnarray}
From (\ref{eq:KKT_Bnk}) we can obtain the relationship  $B_{p,q}^{\star} = B_{n,k}^{\star} - (M-1) \log_2 \left[ \frac{S_{n,k}^{(3)}(1 + S_{p,q}^{(2)})}{S_{p,q}^{(3)}(1 + S_{n,k}^{(2)})} \right]$ for $\forall p \neq n$ and $\forall q \neq k$. Then, the solution of $B^{\star}_{n,k}$ in (\ref{eq:B_nk}) follows by substituting this relationship into (\ref{eq:KKT_lambda}).

\subsection{Derivation of the Lower Bound $R_{n,k}^{LB2}$}\label{appendix:LB2}
Based on (\ref{eq:Rnk_lower1}), we have (\ref{eq:Rnk_lower2_1}) and (\ref{eq:Rnk_lower2_2}) at the top of this page,
\begin{figure*}
\begin{eqnarray}
R_{n,k} &\geq& %\Bbb{E} \left[ \log_2 \left( 1 + \bar{S}_{n,k}^{(1)} \left| \mathbf{h}_{n,k} \mathbf{w}_n \right|^2 + \bar{S}_{n,k}^{(3)}(M-1)^{-1}||\mathbf{h}_{n,k}||^2 \sin^2\theta_{n,k} \right) \right] \nonumber\\
%&&- \Bbb{E} \left[ \log_2 \left( 1 + \bar{S}_{n,k}^{(2)} \left| \mathbf{h}_{n,k} \mathbf{w}_n \right|^2 + \bar{S}_{n,k}^{(3)}(M-1)^{-1}||\mathbf{h}_{n,k}||^2 \sin^2\theta_{n,k} \right) \right] \nonumber\\
%&& - \Bbb{E} \left[ \log_2 \left( 1 + \bar{S}_{n,k}^{(1)} \left| \mathbf{h}_{n,k} \mathbf{w}_n \right|^2 \right) \right]
%+ \Bbb{E} \left[ \log_2 \left( 1 + \bar{S}_{n,k}^{(1)} \left| \mathbf{h}_{n,k} \mathbf{w}_n \right|^2 \right) \right] \nonumber\\
\Bbb{E} \left[ \log_2 \left( 1 + \bar{S}_{n,k}^{(1)} \left| \mathbf{h}_{n,k} \mathbf{w}_n \right|^2 \right) \right]
+ \Bbb{E} \left[ \log_2 \left( 1 + \frac{\bar{S}_{n,k}^{(3)}(M-1)^{-1}||\mathbf{h}_{n,k}||^2 \sin^2\theta_{n,k}}{1 + \bar{S}_{n,k}^{(1)} \left| \tilde{\mathbf{h}}_{n,k} \mathbf{w}_n \right|^2 ||\mathbf{h}_{n,k}||^2} \right) \right] \nonumber\\
&&- \Bbb{E} \left[ \log_2 \left( 1 + \bar{S}_{n,k}^{(2)} \left| \tilde{\mathbf{h}}_{n,k} \mathbf{w}_n \right|^2 ||\mathbf{h}_{n,k}||^2 + \bar{S}_{n,k}^{(3)}(M-1)^{-1}||\mathbf{h}_{n,k}||^2 \sin^2\theta_{n,k} \right) \right] \nonumber\\
\label{eq:Rnk_lower2_1}
&{\geq}& \Bbb{E} \left[ \log_2 \left( 1 + \bar{S}_{n,k}^{(1)} \left| \mathbf{h}_{n,k} \mathbf{w}_n \right|^2 \right) \right]
+ \Bbb{E} \left[ \log_2 \left( 1 + \frac{\bar{S}_{n,k}^{(3)}(M-1)^{-1}||\mathbf{h}_{n,k}||^2 \sin^2\theta_{n,k}}{1 + \bar{S}_{n,k}^{(1)} M^{-1} ||\mathbf{h}_{n,k}||^2} \right) \right] \nonumber\\
&&- \Bbb{E} \left[ \log_2 \left( 1 + \bar{S}_{n,k}^{(2)} M^{-1} ||\mathbf{h}_{n,k}||^2 + \bar{S}_{n,k}^{(3)}(M-1)^{-1}||\mathbf{h}_{n,k}||^2 \sin^2\theta_{n,k} \right) \right] \\
%&=& \Bbb{E} \left[ \log_2 \left( 1 + \bar{S}_{n,k}^{(1)} \left| \mathbf{h}_{n,k} \mathbf{w}_n \right|^2 \right) \right] - \Bbb{E} \left[ \log_2 \left( 1 + \bar{S}_{n,k}^{(1)} M^{-1} ||\mathbf{h}_{n,k}||^2 \right) \right] \nonumber\\
%&&+ \Bbb{E} \left[ \log_2 \left( \frac{1 + \bar{S}_{n,k}^{(1)} M^{-1} ||\mathbf{h}_{n,k}||^2 + \bar{S}_{n,k}^{(3)}(M-1)^{-1}||\mathbf{h}_{n,k}||^2 \sin^2\theta_{n,k}}{1 + \bar{S}_{n,k}^{(2)} M^{-1} ||\mathbf{h}_{n,k}||^2 + \bar{S}_{n,k}^{(3)}(M-1)^{-1}||\mathbf{h}_{n,k}||^2 \sin^2\theta_{n,k}} \right) \right] \nonumber\\
&=& \Bbb{E} \left[ \log_2 \left( 1 + \bar{S}_{n,k}^{(1)} \left| \mathbf{h}_{n,k} \mathbf{w}_n \right|^2 \right) \right] - \Bbb{E} \left[ \log_2 \left( 1 + \bar{S}_{n,k}^{(1)} M^{-1} ||\mathbf{h}_{n,k}||^2 \right) \right] \nonumber\\
\label{eq:Rnk_lower2_2}
&&+ \Bbb{E} \left[ \log_2 \left( 1 + \frac{\frac{d_{n,k}^{-\alpha} P_{n} ||\mathbf{h}_{n,k}||^2}{\sigma_{n,k}^2 K M}}{1 + \bar{S}_{n,k}^{(2)} M^{-1} ||\mathbf{h}_{n,k}||^2 + \bar{S}_{n,k}^{(3)}(M-1)^{-1}||\mathbf{h}_{n,k}||^2 \sin^2\theta_{n,k}} \right) \right].
\end{eqnarray}
\vspace{-4mm}
\hrulefill
\end{figure*}
where (\ref{eq:Rnk_lower2_1}) follows from applying Jensen’s inequality to $| \tilde{\mathbf{h}}_{n,k} \mathbf{w}_n |^2$ and the result $\Bbb{E} [| \tilde{\mathbf{h}}_{n,k} \mathbf{w}_n |^2] = M^{-1}$\cite{Jindal06}.
Then, by applying $\Bbb{E} [\log_2(1+X)] \geq \log_2(1+\frac{1}{\Bbb{E}[1/X]})$ to (\ref{eq:Rnk_lower2_2}), we have
\begin{small}
\begin{eqnarray}\label{eq:R_nk_power}
\hspace{-8mm}&& R_{n,k}\geq\Bbb{E} \left[ \log_2 \left( 1 + \bar{S}_{n,k}^{(1)} \left| \mathbf{h}_{n,k} \mathbf{w}_n \right|^2 \right) \right] \nonumber\\
\hspace{-8mm}&&- \Bbb{E} \left[ \log_2 \left( 1 + \bar{S}_{n,k}^{(1)} M^{-1} ||\mathbf{h}_{n,k}||^2 \right) \right] + \log_2 \bigg( 1 + { \frac{d_{n,k}^{-\alpha} P_{n}} {\sigma_{n,k}^2 K} }
\bigg/ \nonumber\\
\hspace{-8mm}&&{ \Bbb{E} \bigg[ \frac{1 + \bar{S}_{n,k}^{(2)} M^{-1} ||\mathbf{h}_{n,k}||^2 + \bar{S}_{n,k}^{(3)}(M-1)^{-1} \sin^2\theta_{n,k}||\mathbf{h}_{n,k}||^2 }{M^{-1} ||\mathbf{h}_{n,k}||^2 }  \bigg] } \bigg), \nonumber
\end{eqnarray}
\end{small}
which is given by (\ref{eq:R_nk_LB2}) and the final result follows from the approximation of $\Bbb{E}[\sin^2\theta_{n,k}] $ in \cite{Jindal06}. For $\Delta R_{n,k }^{LB2}  $ in (\ref{eq:R_nk_LB2}), we have
\begin{eqnarray}
\hspace{-0.8cm}&&\Delta R_{n,k }^{LB2}  - \Bbb{E} \left[ \log_2 \left( M | \tilde{\mathbf{h}}_{n,k} \mathbf{w}_n |^2 \right) \right] \nonumber\\
%&& \Bbb{E} \left[ \log_2 \left( \frac{1 + \bar{S}_{n,k}^{(1)} | \mathbf{h}_{n,k} \mathbf{w}_n |^2}{1 + \bar{S}_{n,k}^{(1)} M^{-1} ||\mathbf{h}_{n,k}||^2}  \right) \right]
%- \Bbb{E} \left[ \log_2 \left( M | \tilde{\mathbf{h}}_{n,k} \mathbf{w}_n |^2 \right) \right] \nonumber\\
\hspace{-0.8cm}&=& \Bbb{E} \left[ \log_2 \left( \frac{1 + \bar{S}_{n,k}^{(1)} ||\mathbf{h}_{n,k}||^2 | \tilde{\mathbf{h}}_{n,k} \mathbf{w}_n |^2 }{ \left( M + \bar{S}_{n,k}^{(1)} ||\mathbf{h}_{n,k}||^2 \right) | \tilde{\mathbf{h}}_{n,k} \mathbf{w}_n |^2} \right) \right] \nonumber\\
\label{eq:log_M_hw}
\hspace{-0.8cm}&\geq& \Bbb{E} \left[ \log_2 \left( \frac{1 + \bar{S}_{n,k}^{(1)} ||\mathbf{h}_{n,k}||^2 M^{-1} }{ \left( M + \bar{S}_{n,k}^{(1)} ||\mathbf{h}_{n,k}||^2 \right) M^{-1}} \right) \right] = 0, \nonumber
\end{eqnarray}
which follows from Jensen's inequality, since $\log_2 \left( \frac{1+\zeta  x}{(M+\zeta)x} \right) $ is a convex function of $x$.

Define $ g (P_n) \triangleq \Bbb{E} \Big[ \log_2 \Big( 1 + \bar{S}_{n,k}^{(1)} \left| \mathbf{h}_{n,k} \mathbf{w}_n \right|^2 \Big) \Big] - \Bbb{E} \Big[ \log_2 \Big( 1 + \bar{S}_{n,k}^{(1)} M^{-1} ||\mathbf{h}_{n,k}||^2 \Big) \Big]$. It is easy to obtain $ g^{'}(P_n) \triangleq\frac{ {\rm d}g(P_n)}{{\rm d} P_n} = \Bbb{E} \bigg[ \Big( 1 + \frac{d_{n,k}^{-\alpha} k P_n}{\sigma_{n,k}^2 K M} ||\mathbf{h}_{n,k}||^2 \Big)^{-1} \bigg] - \Bbb{E} \bigg[ \Big( 1 + \frac{d_{n,k}^{-\alpha} k P_n}{\sigma_{n,k}^2 K} ||\mathbf{h}_{n,k}||^2 \left| \mathbf{h}_{n,k} \mathbf{w}_n \right|^2 \Big)^{-1} \bigg]$. Applying Jensen’s inequality on the second term, it is easy to show $g^{'}(P_n) \leq 0$. Thus, $g(P_n)$ is a decreasing function of $P_n$ with $ g (0) = 0$. Thus, $ g (P_n) <0$ when $P_n >0$. Moreover, since $ g^{'}(0) = 0$ and $\lim \limits_{P_n \to \infty} g^{'}(P_n)= 0$, $|g(P_n)|$ is relatively small when $P_n$ is small and becomes less sensitive to $P_n$ as $P_n$ goes large. In addition, since $ | \tilde{\mathbf{h}}_{n,k} \mathbf{w}_n |^2 \sim \text{Beta} (1, M-1)$, we have $\Bbb{E} \Big[ \log_2 \left( M | \tilde{\mathbf{h}}_{n,k} \mathbf{w}_n |^2 \right) \Big] = \log_2 (M) + (M-1) \int_0^1 \log_2(x) (1-x)^{M-2} dx =  \log_2(M) -  \log_2(e) \sum_{k=1}^{M-1} k^{-1}  $\cite[4.253.1, 8.365.3]{Gradshteyn2007} . It is easy to show $h(M) = \log_2(M) -  \log_2(e) \sum_{k=1}^{M-1} k^{-1}$ is a increasing function of $M$ and $ \lim_{M \rightarrow \infty} h(M) = -\log_2(e) \mathbf{C} \approx -0.8327$\cite[8.367.2]{Gradshteyn2007}.
%Thus, the term $ g (P_n) $ in (\ref{eq:power_Rnk_mid1}) is no larger than a small constant $-0.8327$. Then, we can conclude that the term $ g (P_n) $ in (\ref{eq:power_Rnk_mid1}) %changes very slowly with $P_n$.

%\subsection{$R_{n,k}$ for $k >1$ is Very Limited with Perfect CSIT}\label{appnedix:Limited_Rate}
%With perfect CSI and $P_{n,k} = P_n / K$, $R_{n,k} $ with $k > 1$ can be upper bound as
%\begin{eqnarray}
%R_{n,k} &=& \log_2 \left( 1 + \frac{\frac{d_{n,k}^{-\alpha} P_{n}}{\sigma_{n,k}^2 K} |\mathbf{h}_{n,k} \mathbf{w}_n|^2}
%{1 + \frac{(k-1) d_{n,k}^{-\alpha} P_{n}}{\sigma_{n,k}^2 K} |\mathbf{h}_{n,k} \mathbf{w}_n|^2} \right) \nonumber \\
%&<& \log_2 \left( 1 + \frac{\frac{d_{n,k}^{-\alpha}}{\sigma_{n,k}^2 K} |\mathbf{h}_{n,k} \mathbf{w}_n|^2}
%{\frac{(k-1) d_{n,k}^{-\alpha}}{\sigma_{n,k}^2 K} |\mathbf{h}_{n,k} \mathbf{w}_n|^2} \right) = \log_2 \left( 1 + \frac{1}{k-1} \right) \nonumber
%\end{eqnarray}
%which is a decreasing function of user index. Moreover, it is easy to see $\sum_{k=2}^K R_{n,k} < \log_2(K)$.

\subsection{Proof of Theorem \ref{theorem:power_alloc}}\label{proof:theorem_power_alloc}
First, denote $C \triangleq P^{\frac{1}{K}} \left( \prod_{p=1}^N \left( 1 - \phi_p \right)^{\frac{1}{N}} \right)\\ \times \left( \prod_{p=1}^N \prod_{q=2}^K \left( \frac{1}{P} + \frac{d_{p,q}^{-\alpha} (q-1) \phi_p}{\sigma_{p,q}^2 K} \right)^{-\frac{1}{NK}}\right)$.
%\begin{eqnarray}\label{eq:defination_C}
%C \triangleq P^{\frac{1}{K}} \left( \prod_{p=1}^N \left( 1 - \phi_p \right)^{\frac{1}{N}} \right) \left( \prod_{p=1}^N \prod_{q=2}^K \left( \frac{1}{P} + \frac{d_{p,q}^{-\alpha} (q-1) \phi_p}{\sigma_{p,q}^2 K} \right)^{-\frac{1}{NK}}\right).
%\end{eqnarray}
Then, by substituting $C$ into (\ref{eq:tilde_R_nk_LB_2}) we can re-express $\tilde{R}^{LB2}_{n,k}  (\left\{\phi_n \right\})$ as (\ref{eq:tilde_R_sum_LB_(1)}) shown at the top of the next page.
\begin{figure*}
\begin{eqnarray}\label{eq:tilde_R_sum_LB_(1)}
\tilde{R}_{sum}^{LB}\; (1)  &=&   \sum_{n=1}^N \log_2 \left( 1 +\frac{ \frac{d_{n,1}^{-\alpha} \phi_n P }{ \sigma_{n,1}^2 K }}
{\frac{M}{(M-1)} + \Gamma \left( \frac{2M-1}{M-1} \right)2^{-\frac{B}{NK(M-1)}}
\left( \prod_{p=1}^N \prod_{q=1}^K \left( \frac{d_{p,q}^{-\alpha}}{\sigma_{p,q}^2} \right)^{\frac{1}{NK}} \right)C}
 \right).
\end{eqnarray}
\vspace{-5mm}
\hrulefill
\end{figure*}
In the following, we identify $C$ as a constant. Then, the problem in (\ref{eq:power_optimization}) becomes a convex optimization problem.
The Lagrangian associated with the problem in (\ref{eq:power_optimization}) is defined as (\ref{eq:mathcal_L_phi_nk}) $\forall n, k$ at the top of this page,
\begin{figure*}
\begin{equation}\label{eq:mathcal_L_phi_nk}
\mathcal{L}_{\phi}(\phi_{n,k}, \lambda_{\phi}) = -\prod_{n=1}^N \left( 1 +\frac{ \frac{d_{n,1}^{-\alpha} \phi_n P }{ \sigma_{n,1}^2 K }}{\frac{M}{(M-1)} + \Gamma \left( \frac{2M-1}{M-1} \right)2^{-\frac{B}{NK(M-1)}} \left( \prod_{p=1}^N \prod_{q=1}^K \left( \frac{d_{p,q}^{-\alpha}}{\sigma_{p,q}^2} \right)^{\frac{1}{NK}} \right)C} \right)
+ \lambda_{\phi} \left( \sum_{n=1}^{N} \phi_n - 1 \right).
\end{equation}
\vspace{-5mm}
\hrulefill
\end{figure*}
where $\lambda_{\phi}$ is the Lagrangian multiplier associated with the sum power constraint. According to the KKT conditions for optimality, we can obtain (\ref{eq:KKT_phi_nk}) and (\ref{eq:KKT_lambda_phi_nk}) at the top of this page.
\begin{figure*}
\begin{eqnarray}\label{eq:KKT_phi_nk}
\frac{\partial \mathcal{L}_{\phi} (\phi_n, \lambda_{\phi})}{\partial \phi_n}\Bigg |_{\left\{\phi_n^{\star} \right\}, \lambda_{\phi}^{\star}} &=& -\frac{ \frac{d_{n,1}^{-\alpha} P }{ \sigma_{n,1}^2 K } } {\frac{M}{(M-1) } + \Gamma \left( \frac{2M-1}{M-1} \right)
\left( \prod_{p=1}^N \prod_{q=1}^K \left( \frac{d_{p,q}^{-\alpha}}{\sigma_{p,q}^2} \right)^{\frac{1}{NK}} \right) 2^{-\frac{B}{NK(M-1)}} C} \nonumber \\
&& \hspace{-0.5cm} \times \prod_{i=1,i \neq n}^N \left( 1 + \frac{\frac{d_{i,1}^{-\alpha} P }{ \sigma_{i,1}^2 K }  \phi_i^{\star} }
{\frac{M}{(M-1)} + \Gamma \left( \frac{2M-1}{M-1} \right) \left( \prod_{p=1}^N \prod_{q=1}^K \left(\frac{d_{p,q}^{-\alpha}}{\sigma_{p,q}^2} \right)^{\frac{1}{NK}} \right) 2^{-\frac{B}{NK(M-1)}} C} \right)  + \lambda_{\phi}^{\star} = 0, ~~~~~ \\
\frac{\partial \mathcal{L}_{\phi} (\phi_n, \lambda_{\phi})}{\partial \lambda_{\phi}} \Bigg |_{\left\{\phi_n^{\star} \right\}}&=&
\sum_{n=1}^{N} \phi_n^{\star} - 1 = 0. \label{eq:KKT_lambda_phi_nk}
\end{eqnarray}
\vspace{-4mm}
\hrulefill
\end{figure*}

According to (\ref{eq:KKT_phi_nk}),  we can obtain the following relationship for $\forall i \neq n$ that,
\begin{eqnarray}
\hspace{-7mm}&&\phi_i^{\star} = \phi_n^{\star} + K \Bigg( \frac{M}{(M-1)P} + \Gamma \left( \frac{2M-1}{M-1} \right) \nonumber\\
\hspace{-7mm}&&\hspace{1mm}\times\left( \prod_{p=1}^N \prod_{q=1}^K \Big( \frac{d_{p,q}^{-\alpha}}{\sigma_{p,q}^2} \Big)^{\frac{1}{NK}} \right) \frac{2^{-\frac{B}{NK(M-1)}}}{P} C  \Bigg)
\bigg( \frac{\sigma_{n,1}^2}{d_{n,1}^{-\alpha}} - \frac{\sigma_{i,1}^2}{d_{i,1}^{-\alpha}} \bigg).\nonumber
\end{eqnarray}
Substituting the above relationship into (\ref{eq:KKT_lambda_phi_nk}), we can obtain (\ref{eq:phi_n_res_final}). Substituting $\phi_n^{\star}$ as a function of $C$ given by (\ref{eq:phi_n_res_final}) into the definition of $C$ above, we can obtain the equation about the optimal $C^{\star}$ given by (\ref{eq:equation_C}).
%\begin{eqnarray}
%\phi_n^{\star} &=& \frac{1}{N} - \frac{K}{N} \left(
%\frac{M}{(M-1)P} + \Gamma \left( \frac{2M-1}{M-1} \right) \left( \prod_{p=1}^N \prod_{q=1}^K \Big( \frac{d_{p,q}^{-\alpha}}{\sigma_{p,q}^2} \Big)^{\frac{1}{NK}} \right) 2^{-\frac{B}{NK(M-1)}} \frac{C}{P}
%\right) \nonumber\\
%&&\times \sum_{i=1,i \neq n}^N \left( \frac{\sigma_{n,1}^2}{d_{n,1}^{-\alpha}} - \frac{\sigma_{i,1}^2}{d_{i,1}^{-\alpha}} \right).
%\end{eqnarray}
Moreover, according to (\ref{eq:cluster_criteria1}) and (\ref{eq:cluster_criteria2}), it is easy to prove $\sum_{i=1,i \neq m}^N \left( \frac{\sigma_{m,1}^2}{d_{m,1}^{-\alpha}} - \frac{\sigma_{i,1}^2}{d_{i,1}^{-\alpha}} \right) \leq \sum_{i=1,i \neq n}^N \left( \frac{\sigma_{n,1}^2}{d_{n,1}^{-\alpha}} - \frac{\sigma_{i,1}^2}{d_{i,1}^{-\alpha}} \right)$, $\forall  1<m < n <N$. Then, the order of the coefficients $\phi_i$s follows.

\bibliographystyle{IEEEtran}
\bibliography{sunliang_bib}

\end{document}